\documentclass[10pt, twocolumn]{IEEEtran}

\usepackage[normalem]{ulem}

\usepackage{amsmath}
\usepackage{amsthm}
\usepackage{amssymb}
\usepackage{verbatim}
\usepackage{algorithm}
\usepackage{algorithmic}
\usepackage{graphicx}
\usepackage{paralist}
\usepackage{color}
\usepackage{url}

\usepackage{xcolor,colortbl}
\usepackage{verbatim}
\usepackage{algorithm}
\usepackage{algorithmic}
\usepackage{graphicx}
\usepackage{epstopdf}  
\usepackage{paralist}
\usepackage{color}
\usepackage{subfigure}
\usepackage{ulem}

\usepackage{multirow}

\newtheorem{lemma}{Lemma}
\newtheorem{theorem}{Theorem}
\newtheorem{definition}{Definition}

\floatname{algorithm}{Algorithm}
\newtheorem*{remark}{Remark}

\title{Finding Needles in a Haystack: Missing Tag Detection in Large RFID Systems}
\author{Jihong Yu, Lin Chen, Kehao Wang \thanks{J. Yu and L. Chen are with Lab. Recherche Informatique (LRI-CNRS UMR 8623), Univ. Paris-Sud, 91405 Orsay, France, \{jihong.yu, chen\}@lri.fr. K. Wang is with Dept. Inform. Eng., Wuhan University of Technology, China, kehao.wang@whut.edu.cn.}}

\begin{document}

\maketitle

\begin{abstract}
Radio frequency identification (RFID) technology has been widely used in missing tag detection to reduce and avoid inventory shrinkage. In this application, promptly finding out the missing event is of paramount importance. However, existing missing tag detection protocols cannot efficiently handle the presence of a large number of unexpected tags whose IDs are not known to the reader, which shackles the time efficiency. To deal with the problem of detecting missing tags in the presence of unexpected tags, this paper introduces a two-phase Bloom filter-based missing tag detection protocol (BMTD). The proposed BMTD exploits Bloom filter in sequence to first deactivate the unexpected tags and then test the membership of the expected tags, thus dampening the interference from the unexpected tags and considerably reducing the detection time. Moreover, the theoretical analysis of the protocol parameters is performed to minimize the detection time of the proposed BMTD and achieve the required reliability simultaneously. Extensive experiments are then conducted to evaluate the performance of the proposed BMTD. The results demonstrate that the proposed BMTD significantly outperforms the state-of-the-art solutions.

\end{abstract}


\section{Introduction}

\subsection{Background}
Recent years have witnessed an unprecedented development of the radio frequency identification (RFID) technology. As a promising low-cost technology, RFID is widely utilized in various applications ranging from inventory control~\cite{DoD2004}~\cite{DoD2007}~\cite{bu2012misplaced}, supply chain management and logistics~\cite{lee2008supply}~\cite{sheng2008finding}~\cite{chen2011efficient}~\cite{qiao2011polling}~\cite{liu2014query}
to tracking/location \cite{ni2011tracking}~\cite{yang2013localization}~\cite{han2014twins}. In these applications, an RFID system typically consists of one or several RFID readers and a large number of RFID tags. Specially, the RFID reader is a device equipped with a dedicated power source and an antenna and can collect and process the information of tags within its coverage area. An RFID tag, on the other hand, is a low-cost microchip labeled with a unique serial number (ID) to identify an object and can receive and transmit the radio signals via the wireless channel.
More specifically, the tags are generally classified into two categories: passive and active tags. Passive tags are energized by the radio wave of the reader, while active tags have power sources and relatively long communication range.

\subsection{Motivation and problem statement}
According to the statistics presented in~\cite{National2015}, inventory shrinkage, a combination of shoplifting, internal theft, administrative and paperwork error, and vendor fraud, resulted in 44 billion dollars in loss for retailers in $2014$.
Fortunately, RFID technology can be used to reduce the cost by monitoring products for its low cost and non-line-of-sight communication pattern.
Obviously, the first step in the application of loss prevention is to determine whether there is any missing tag. Hence, quickly finding out the missing tag event is of practical importance.

The presence of unexpected tags, however, prolongs the detection time and even leads to miss detection.
Here, we present two examples to motivate the presence of unexpected tags in realistic scenarios.
\begin{itemize}
\item \emph{Example 1.} Consider a retail store with expensive goods and a much larger amount of inexpensive goods, and an RFID system is deployed to monitor the goods. Because of the higher value of expensive products, they are expected to be detected more frequently, but the tags of inexpensive goods also response the interrogation of readers, which influences the decision of readers.
\item \emph{Example 2.} Consider a large warehouse rented to multiple companies where the products of the same company may be placed in different zones according to their individual categories, such as child food and adult food, chilled food and ambient food. When detecting the tags identifying products from one company, readers also receive the feedbacks from the tags of other companies.
\end{itemize}

In both examples, how to effectively reduce the impact of unexpected tags is of critical importance in missing tag detection. In this paper, we consider a scenario, as depicted in Fig.~\ref{Fig:single_reader}, where each product is affixed by an RFID tag. The reader stores the IDs of expected tags. The problem we address is how to detect missing expected tags in the presence of a large number of unexpected tags in the RFID systems in a reliable and time-efficient way.

\begin{figure}[htbp]
\vspace{-0cm}
\centering
\includegraphics[width=8cm]{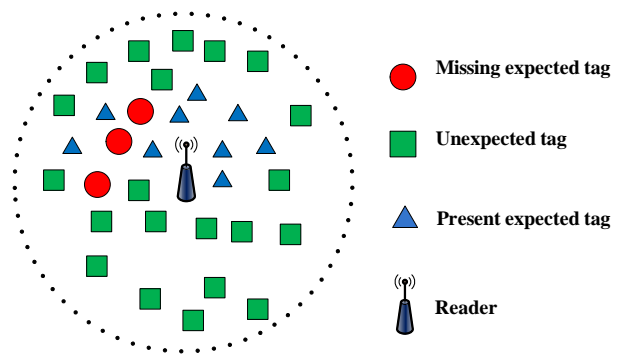}
\vspace{-0cm}
\caption{Missing tag detection with the presence of unexpected tags.}
\label{Fig:single_reader}
\end{figure}

\subsection{Prior art and limitation}
Prior related work can be classified into three categories from the perspective of detecting missing tags: missing tag detection protocols, tag identification protocols, and tag estimation protocols.

There are two types of missing tag detection protocols: probabilistic~\cite{tan2008monitor}~\cite{luo2012probabilistic}~\cite{luo2014missing}~\cite{shahzad2015expecting} and deterministic~\cite{li2010identifying}~\cite{zhang2011fast}~\cite{liu2015completely}. The probabilistic protocols find out a missing tag event with a certain required probability if the number of missing tags exceeds a given threshold, thus they are more time-efficient but return weaker results in comparison with the deterministic protocols that report all IDs of the missing tags. Actually, they can be used together such that a probabilistic protocol is executed in the first phase as an alarm that reports the absence of tags and then a deterministic protocol is executed to report IDs of missing tags. Unfortunately, all missing tag detection protocols except RUN~\cite{shahzad2015expecting} work on the hypothesis of a perfect environment without unexpected tags and thus fail to effectively detect missing tags in the presence of unexpected tags. Although RUN~\cite{shahzad2015expecting} is tailored for missing tag detection in the RFID systems with unexpected tags, all unexpected tags may always participate in the interrogation, which leads to the significant degradation of the performance when the unexpected tag population size scales.

Tag identification protocols~\cite{myung2006adaptive}~\cite{namboodiri2010energy}~\cite{la2011anticollision}~\cite{shahzad2013probabilistic} can identify all tags in the interrogation region. To detect missing tags, tag identification protocols can be executed to obtain the IDs of the tags present in the population and then the missing tags can be found out by comparing the collected IDs with those recorded in the database. However, they are usually time-consuming~\cite{li2010identifying} and fail to work when it is not allowed to read the IDs of tags due to privacy concern.

Tag estimation protocols~\cite{qian2011cardinality}~\cite{shahzad2012everybit}~\cite{zheng2013zoe}~\cite{chen2013understanding} are used to estimate the number of tags in the interrogation region. If many expected tags are absent in RFID systems without unexpected tags, a missing tag event may be detected by comparing the estimation and the number of expected tags stored in the database. However, the estimation error may be misinterpreted as missing tags and cause detection error, especially when there are only a few missing tags. Moreover, the estimation protocol cannot handle the case with a large number of unexpected tags.

\subsection{Proposed solution and main contributions}
Motivated by the detrimental effects of unexpected tags on the performance of missing tag detection, we devise a reliable and time-efficient protocol named \underline{B}loom filter-based \underline{m}issing \underline{t}ag \underline{d}etection protocol (BMTD). Specifically, BMTD consists of two phases, each consisting of a number of rounds.
\begin{itemize}
\item In each round of the first phase, the reader fist constructs a Bloom filter by mapping all the expected tag IDs into it such that each tag has multiple representative bits. Then the constructed Bloom filter is broadcasted to all tags. If at least one representative bit of a tag is '0's, it finds itself unexpected and will not participate in the rest of BMTD. Thus, the number of active unexpected tags is considerably reduced.
\item Subsequently, in each round of the second phase, the reader constructs a Bloom filter by aggregating the feedbacks from the remaining tags and uses it to check whether any expected tag is absent from the population.
\end{itemize}

The major contributions of this paper can be articulated as follows.
First, we propose a new solution for the important and challenging problem of missing tag detection in the presence of a large number of unexpected tags by employing Bloom filter to filter out the unexpected tags and then detect the missing tags.
Second, we perform the theoretical analysis for determining the optimal parameters used in BMTD that minimize the detection time and also meet the required reliability.
Third, we perform extensive simulations to evaluate the performance of BMTD. The results show that BMTD significantly outperforms the state-of-the-art solutions.

The remainder of the paper is organised as follows.
Section \ref{sec:related} gives a brief overview of related work.
In Section \ref{sec:model and formulation}, we formally present the missing tag detection problem and describe the design goal and requirements.
In Section \ref{sec:protocol} and \ref{sec:parameter}, we elaborate the designed protocol and perform the theoretical analysis of the parameter configuration, respectively.
In Section \ref{sec:estimation}, we introduce the method to estimate the unexpected tag population size.
Then the extensive simulations are conducted in
Section \ref{sec:simulation}.
Finally, we conclude our paper in Section~\ref{sec:conclusion}.

\section{Related Work}
\label{sec:related}

Extensive research efforts have been devoted to detecting missing tags by using probabilistic method~\cite{tan2008monitor}~\cite{luo2012probabilistic}~\cite{luo2014missing}~\cite{shahzad2015expecting} and deterministic method~\cite{li2010identifying}~\cite{zhang2011fast}~\cite{liu2015completely}. Next, we briefly review the existing solutions of missing tag detection problem.

\subsection{Probabilistic protocols}
The objective of probabilistic protocols is to detect a missing tag event with a predefined probability. Tan \textit{et al.} initiate the study of probabilistic detection and propose a solution called TRP in~\cite{tan2008monitor}. TRP can detect a missing tag event by comparing the pre-computed slots with those picked by the tags in the population. Different from our BMTD, TRP does not take into account the negative impact of unexpected tags. Follow-up works~\cite{luo2012probabilistic}~\cite{luo2014missing} employ multiple seeds to increase the probability of the singleton slot. Same to TRP, they are required to know all the tags in the population. The latest probabilistic protocol called RUN is proposed in~\cite{shahzad2015expecting}. The difference with previous works lies in that RUN considers the influence of unexpected tags and can work in the environment with unexpected tags. However, RUN does not eliminate the interference of unexpected tags fundamentally such that the false positive probability does not decrease with respect to the unexpected tag population size, which shackles the detection efficiency especially in the presence of a large number of unexpected tags. In addition, the first frame length is set to the double of the cardinality of the expected tag set in RUN, which is not established by theoretical analysis and leads to the failure of estimation method in RUN when the number of the unexpected tags is far larger than that of the expected tags.

\subsection{Deterministic protocols}
The objective of deterministic protocols is to exactly identify which tags are absent. Li \textit{et al.} develop a series of protocols in~\cite{li2010identifying} which intend to reduce the radio collision and identify a tag not in the ID level but in the bit level. Subsequently, Zhang \textit{et al.} propose another series of determine protocols in~\cite{zhang2011fast} of which the main idea is to store the bitmap of tag responses in all rounds and compare them to determine the present and absent tags. But how to configure the protocol parameters is not theoretically analyzed. More recently, Liu \textit{et al.}~\cite{liu2015completely} enhance the work by reconciling both 2-collision and 3-collision slots and filtering the empty slots such that the time efficiency can be improved. None of existing deterministic protocols, however, have been designed to work in the chaotic environment with unexpected tags.

\section{System Model and Problem Formulation}
\label{sec:model and formulation}

\subsection{System model}

Consider a large RFID system consisting of a single RFID reader and a large number of RFID tags.
The reader broadcasts the commands and collects the feedbacks from the tags. In the RFID system, the tags can be either battery-powered active ones or lightweight passive ones that are energized by radio waves emitted from the reader. In this paper, we first take account of the single-reader case and then extend the proposed protocol to the multi-reader case.

The communications between the readers and the tags follow the \textit{Listen-before-talk} mechanism~\cite{han2010counting}: A reader initiates communication first by sending commands and broadcasting the parameters to tags, such as the frame size, random seeds, and then each tag responds in its chosen time slot. Consider an arbitrary time slot, if no tag replies in this slot, it is called an \textit{empty slot}; otherwise, it is called a \textit{nonempty slot}. Only one bit is needed to distinguish an empty slot from a nonempty slot: `0' for an empty slot with an ideal channel while `1' for a nonempty slot with a busy channel.

During the communications, the tag-to-reader transmission rate and the reader-to-tag transmission rate may differ with each other and are subject to the environment. In practice, the former can be either $40$kb/s $\sim$ $640$kb/s in the FM0 encoding format or $5$kb/s $\sim$ $320$kb/s in the modulated subcarrier encoding format, while the later is normally about $26.7$kb/s $\sim$ $128$kb/s~\cite{C1G22005}.

\subsection{Problem formulation}

In the considered RFID system, we use $\mathbb{E}$ to denote the set of IDs of the expected tags which are expected to be present in a population and target tags to be monitored. In the RFID system, we assume that an unknown number of tags, $m$, out of these $|\mathbb{E}|$ tags are missing. Note that $|\cdot |$ stands for the cardinality of a set. Denote by $\mathbb{E}_r$ the set of IDs of the remaining $|\mathbb{E}|-m$ tags that are actually present in the population. Let $\mathbb{U}$ be the set of IDs of unexpected tags within the interrogation region of the reader which does not need to be monitored. The reader may neither knows exactly the IDs of unexpected tags nor does it know the cardinality of $\mathbb{U}$.

Let $M$ be a threshold on the number of missing expected tags.
We use $P_{sys}$ to denote the probability that the reader can detect a missing event. The optimum missing tag detection problem is formally defined as follows.
\begin{definition}[Optimum missing tag detection problem]
Given $|\mathbb{U}|$ unexpected tags where both $|\mathbb{U}|$ and the IDs of tags in $\mathbb{U}$ are unknown, the optimum missing tag detection problem is to devise a protocol of minimum execution time capable of detecting a missing event with probability $P_{sys} \ge \alpha$ if $m \ge M$, where $\alpha$ is the system requirement on the detection reliability.
\end{definition}

Table~\ref{Tab:notation} summaries the main notations used in the paper.
\begin{table}[!htbp]
\centering
\caption{Main Notations}
\label{Tab:notation}
\begin{tabular}{|l|l|}
\hline
\textbf{Symbols} & \textbf{Descriptions}\\
\hline
$\mathbb{E}$& set of target tags that need to be monitored \\
\hline
$\mathbb{E}_r$& tags that are actually present in the population\\
\hline
$\mathbb{U}$& set of unexpected tags\\
\hline
$\alpha$& required detection reliability\\
\hline
$m$& number of missing expected tags\\
\hline
$M$& threshold to detect missing tags\\
\hline
$P_{sys}$& prob. of detecting a missing event in BMTD\\
\hline
$J$& number of rounds in Phase 1\\
\hline
$l_j$& length of the $j$-th frame of Phase 1 \\
\hline
$k_j$& number of hash functions in the $j$-th frame of Phase 1\\
\hline
$s_j$& random seed used in the $j$-th frame of Phase 1\\
\hline
$\mathbb{U}_r$ & set of remaining active unexpected tags after Phase 1\\
\hline
$N^*$ & number of remaining active tags after Phase 1\\
\hline
$P_{1,j}$& false positive rate in the $j$-th frame of Phase 1\\
\hline
$T_1$& time cost of Phase 1\\
\hline
$W$&number of rounds in Phase 2\\
\hline
$f_w$& length of the $w$-th frame of Phase 2 \\
\hline
$R_w$& number of hash functions in the $w$-th frame of Phase 2\\
\hline
$d_w$& random seed used in the $w$-th frame of Phase 2\\
\hline
$P_{2,w}$& false positive rate in the $w$-th frame of Phase 2\\
\hline
$T_2$& time taken to execute $W$ rounds in Phase 2\\
\hline
$T$& theoretical execution time\\
\hline
$q$& prob. of detect a missing tag in a given slot of Phase 2\\
\hline
$Z$& random variable for slot of the first detection\\
\hline
$E[T_D]$& expected detection time of BMTD\\
\hline
\end{tabular}
\end{table}

\section{Bloom Filter-based Missing Tag Detection Protocol}
\label{sec:protocol}

\begin{figure*}[htbp]
\centering
\subfigure[Phase 1: unexpected tag deactivation]{
\begin{minipage}[t]{0.49\linewidth}
\centering
\includegraphics[width=0.748\textwidth]{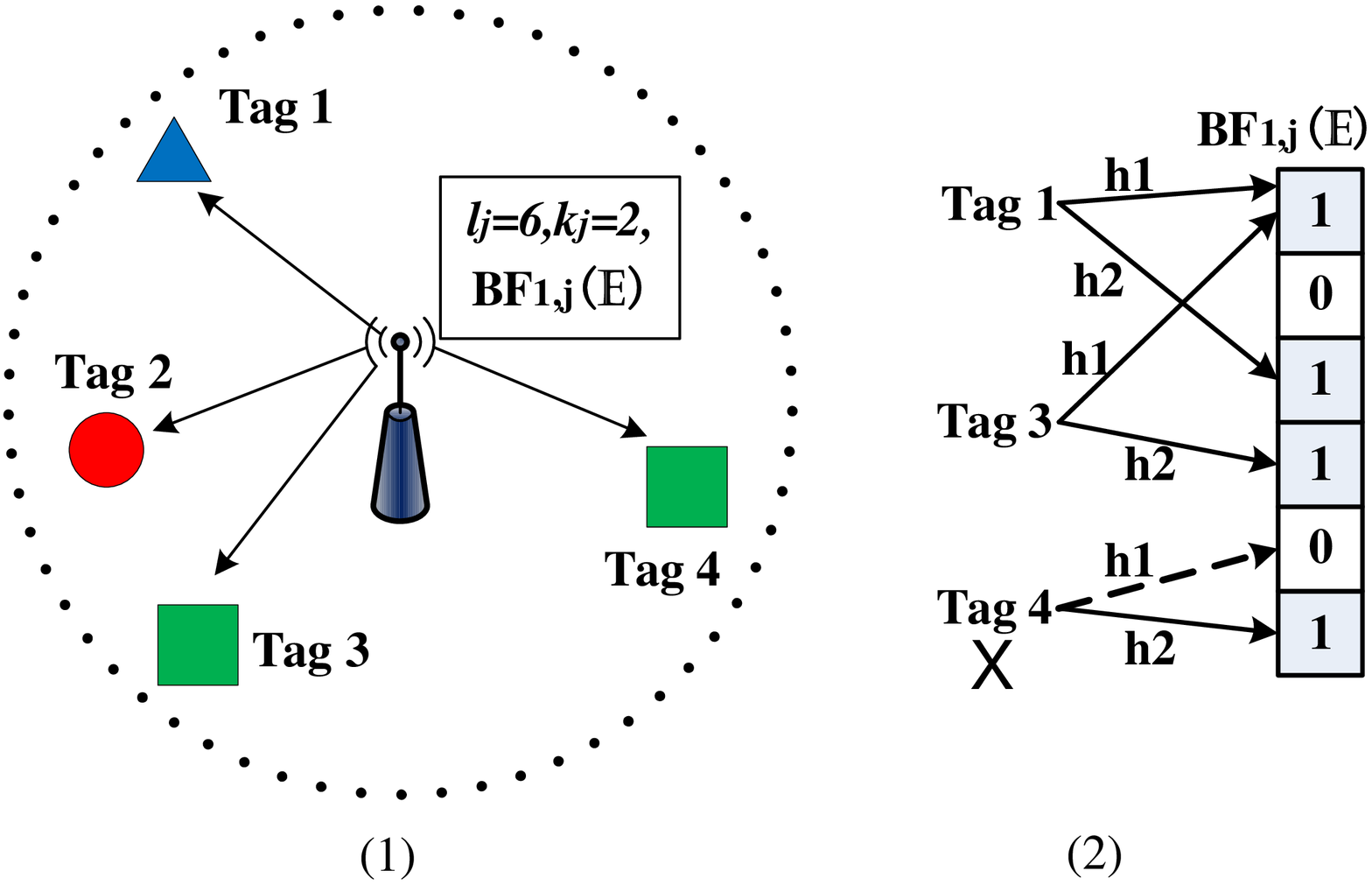}
\label{Fig:phase1}
\end{minipage}}
\subfigure[Phase 2: missing tag detection]{
\begin{minipage}[t]{0.49\linewidth}
\centering
\includegraphics[width=1\textwidth]{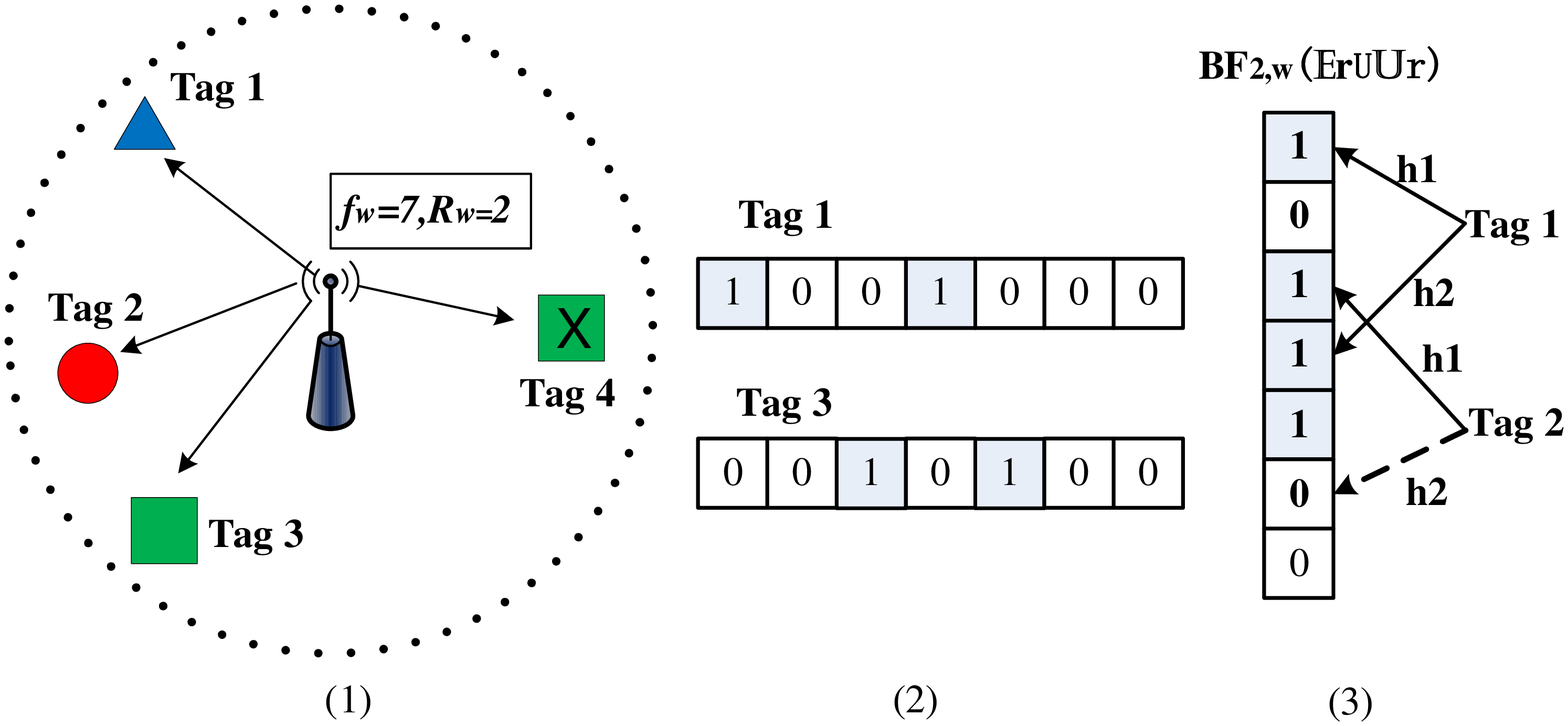}
\label{Fig:phase2}
\end{minipage}}
\caption{Example illustrating BMTD}
\end{figure*}

\subsection{Design rational and protocol overview}
To improve the time efficiency of detecting missing tags in the presence of a large number of unexpected tags in the population, we limit the interference of unexpected tags in our protocol. To achieve this goal, we employ a powerful technique called \textit{Bloom filter} which is a space-efficient probabilistic data structure for representing a set and supporting set membership queries~\cite{bloom1970space} to rule out the unexpected tags in the set $\mathbb{U}$, which efficiently reduces their interference and thus the overall execution time. Following this idea, we propose a \textit{Bloom filter-based Missing Tag Detection protocol} (BMTD), by which Bloom filters are sequentially constructed by the reader and by the feedbacks from the active tags in the RFID system.

The BMTD consists of two phases: 1) the unexpected tag deactivation phase and 2) the missing tag detection phase.
\begin{itemize}
\item The first phase is divided into $J$ rounds where the reader constructs $J$ Bloom filters by mapping the recorded IDs in the reader to deactivate the unexpected tags after identifying them.
\item The second phase is divided into $W$ rounds. The reader constructs $W$ Bloom filters according to the responses of the remaining active tags and uses the Bloom filters to detect any missing event. Our protocol either detects a missing event or reports no missing event if the reader does not detect a missing event after $W$ rounds.
\end{itemize}

We elaborate the design of the BMTD in the rest of this section.

\subsection{Phase 1: unexpected tag deactivation}

In Phase 1, we use Bloom filters to reduce the number of active unexpected tags. Specifically, in the $j$-th round of Phase 1 ($j=1,2,...,J$), the reader first constructs a Bloom filtering vector by mapping the expected tags in set $\mathbb{U}$ into an $l_j$-bit array using $k_j$ hash functions with random seed $s_j$. Here, we denote the $l_j$-bit Bloom filter vector as $BF_{1,j} (\mathbb{E})$. How the values of $l_j$, $k_j$ are chosen and how $J$ is calculated are analysed in Sec.~\ref{sec:parameter} on parameter optimisation.

Then, the reader broadcasts the $l_j$-bit Bloom filtering vector, $k_j$ and $s_j$ to all tags. Upon receiving $BF_{1,j} (\mathbb{E})$, $k_j$, and $s_j$, each tag maps its ID to $k_j$ bits pseudo-randomly at positions $h_1(ID), h_2(ID),\cdots, h_{k_j}(ID)$, and checks the corresponding positions in $BF_{1,j} (\mathbb{E})$. If all of $k_j$ bits are $1$, then the tag regards itself expected by the reader. If any of $k_j$ bits is $0$, the tag regards that it is unexpected and then remains silent in the rest of the time.

Let $\mathbb{U}_{j}$ denote the set of the remaining active unexpected tags after the $j$-th round of Phase 1, and let $\mathbb{U}_j \cap BF_{1,j} (\mathbb{E})$ denote the set of unexpected tags that pass the membership test of $BF_{1,j} (\mathbb{E})$. Since the Bloom filter has no false negatives, the set of remaining active tags can be represented as $\mathbb{E}_r \cup \mathbb{U}_{j-1} \cap BF_{1,j} (\mathbb{E})$.

After $J$ rounds when Phase 1 is terminated, the number of remaining active unexpected tags, termed as $|\mathbb{U}_{r}|$, is $|\mathbb{U}_{J} \cap BF_{1,J} (\mathbb{E})|$. The present tag population size can be written as $|\mathbb{E}_r \cup \mathbb{U}_{r}|$. Subsequently, the reader enters Phase 2.

\subsection{Phase 2: missing tag detection}
In the second phase, we still employ Bloom filter to detect a missing tag event.
Note that the parameters that the reader broadcasts in each round in Phase 2 except random seeds are identical.
In the $w$-th round of Phase 2 ($w=1,2,...,W$), the reader first broadcasts the parameters containing the Bloom filter size $f_w$, the number of hash functions $R_w$, and a new random seed $d_w$. How their values are chosen and how $W$ is calculated are analysed in Sec.~\ref{sec:parameter} on parameter optimisation.

After receiving the configuration parameters, each tag in the set $\mathbb{E}_r \cup \mathbb{U}_{r}$ selects $R_w$ slots at the indexes $h_v(ID)$ ($1 \le v \le R_w$) in the frame of $f_w$ slots and transmits a short response at each of the $R_w$ corresponding slots. As a consequence, a Bloom filter is formed in the air by the responses from the remaining active tags. In each round, there are two types of slots: empty slots and nonempty slots.

According to the responses from the tags, the reader encodes an $f_w$-bit Bloom filter as follows: If the $i$-th slot is empty, the reader sets $i$-th bit of the $f_w$-bit vector to be '0', otherwise '1'. Consequently, a virtual Bloom filter is constructed using which the reader then performs membership test. Let $BF_{2,w}(\mathbb{E}_r \cup \mathbb{U}_{r})$ denote the constructed Bloom filter in $w$-th round.

To perform membership test, the reader uses tag IDs from the expected tag set $\mathbb{E}$. Specifically, for each ID in $\mathbb{E}$, the reader maps it into $R_w$ bits at positions $h_v(ID)$ ($1 \le v \le R_w$) in $BF_{2,w}(\mathbb{E}_r \cup \mathbb{U}_{r})$. If all of them are '1's, then the tag is regarded as present. Otherwise, the tag is considered to be missing. If a missing event is detected in $w$-round, the reader terminates the protocol without executing the remaining rounds. Otherwise, the reader initiates a new round until the protocol runs $W$ rounds. If the reader does not detect a missing event after $W$ rounds, it reports no missing event, i.e., the number of missing tags $m$ is less than the threshold $M$.


\subsection{An illustrative example of BMTD}
We present an illustrative example to show the execution of BMTD. Consider an RFID system with $4$ tags. We assume that the reader needs to monitor tag $1$ and tag $2$ and thus knows their IDs, i.e., $\mathbb{E}$$=$$\{$ID1, ID2$\}$, but it is not aware of the presence of tag $3$ and tag $4$, who are unexpected, i.e., $\mathbb{U}$$=$$\{$ID3, ID4$\}$. In the example, tag $2$ is missing from the population.

As shown in (1) of Fig.~\ref{Fig:phase1}, the reader first constructs a Bloom filter $BF_{1,j} (\mathbb{E})$ by mapping IDs in $\mathbb{E}$ and broadcasts a message containing $BF_{1,j} (\mathbb{E})$ and the values of $k_j$ and $l_j$. Here we assume $J=1$, $k_j=2$ and $l_j=6$. After receiving $BF_{1,j} (\mathbb{E})$, each tag checks if it is an expected tag. As shown in (2) of Fig.~\ref{Fig:phase1}, tag $1$ finds itself expected due to the fact that both $h_1(\text{ID}1)$ and $h_2(\text{ID}1)$ are equal to $1$. However, tag $4$ realizes that it is unexpected for $h_1(\text{ID}4)=0$ and deactivates itself. Different from tag $4$, actually unexpected tag $3$ passes the test and will participate in the rest of BMTD.

As depicted in (1) of Fig.~\ref{Fig:phase2}, after the first phase, the reader starts to detect missing tags by broadcasting parameters $f_w$ and $R_w$. Here we assume $W=1$, $R_w=2$ and $f_w=7$. By using $f_w$ and $R_w$, tag $1$ and tag $3$ generate a Bloom filter vector, respectively, which is shown in (2) of Fig.~\ref{Fig:phase2}. Then they transmit following their individual Bloom filter vector. By sensing the channel, the reader can encode a Bloom filter and use it to check the IDs in $\mathbb{E}$ one by one. As shown in (3) of Fig.~\ref{Fig:phase2}, since the Bloom filter is constructed based on the responses of tag $1$ and tag $3$, tag $1$ passes the test but tag $2$ fails and is regarded as absent. 
Then the protocol reports a missing event.

\section{Performance optimisation and parameter tuning}
\label{sec:parameter}

In this section, we investigate how the parameters in the BMTD are configured to minimise the execution time while ensuring the performance requirement.

\subsection{Tuning parameters in Phase 1}

According to the property of Bloom filter, false negatives are impossible. The false positive rate of the Bloom filter $BF_{1,j} (\mathbb{E})$ in the $j$-th round in Phase 1, defined as $P_{1,j}$, can be calculated as follows~\cite{bloom1970space}:
\begin{equation}
P_{1,j} = \left[1-\left(1-\frac{1}{l_j}\right)^{|\mathbb{E}| k_j}\right]^{k_j} \approx (1-e^{-{|\mathbb{E}| k_j}/l_j})^{k_j}.
\label{eq:p_1j}
\end{equation}

By rearranging~\eqref{eq:p_1j}, we can express the Bloom filter size in the $j$-th round as
\begin{equation}
l_j = \frac{-|\mathbb{E}| k_j}{\ln (1- P^{\frac{1}{k_j}}_{1,j})}.
\end{equation}
The total time spent in this round can thus be calculated as $l_j * t_r$, where $t_r$ denotes the per bit transmission time from reader to tags.

We denote $C_j$ the cost to detect and deactivate an unexpected tag as follows:
\begin{equation}
C_j = \frac{l_j t_r}{|\mathbb{U}|(1-P_{1,j})}
= \frac{-t_r |\mathbb{E}| k_j}{|\mathbb{U}|(1-P_{1,j}) \ln (1- P^{\frac{1}{k_j}}_{1,j})}.
\end{equation}
From the expression of $C_j$, it can be noted that $C_j$ represents the average time consumed to detect and deactive an unexpected tag in the $j$-th round. In our design we minimize $C_j$ so as to achieve the optimal time-efficiency. To minimize $C_j$, we first compute the derivative of $C_j$ with respect to $k_j$ as follows:
\begin{equation}
\frac{\mathbf{d} C_j}{\mathbf{d}k_j}
= \frac{|\mathbb{E}| t_r \left(P^{\frac{1}{k_j}}_{1,j} \ln P_{1,j}
  - k_j (1-P^{\frac{1}{k_j}}_{1,j}) \ln(1-P^{\frac{1}{k_j}}_{1,j})\right)} {|\mathbb{U}|(1-P_{1,j}) k_j (1-P^{\frac{1}{k_j}}_{1,j}) \ln^2(1-P^{\frac{1}{k_j}}_{1,j})}.
\end{equation}
Furthermore, let $\frac{\mathbf{d} C_j}{\mathbf{d}k_j}=0$, we can obtain
\begin{equation}
P^{\frac{1}{k_j}}_{1,j}=\frac{1}{2},
\label{Eq:p_k}
\end{equation}
and the unique minimiser $k_j^* = \frac{-\ln P_{1,j}}{\ln 2}$ as $\frac{\mathbf{d} C_j}{\mathbf{d}k_j}>0$ when $k_j > \frac{-\ln p_{1,j}}{\ln 2}$, and $\frac{\mathbf{d} C_j}{\mathbf{d}k_j}<0$ when $k_j < \frac{-\ln p_{1,j}}{\ln 2}$. Therefore, $C_j$ reaches the minimum value when $P^{\frac{1}{k_j^*}}_{1,j}=\frac{1}{2}$. The optimum Bloom filter size, denoted as $l^*_j$, can be computed as
\begin{equation}
l^*_j = \frac{|\mathbb{E}| k_j^*}{\ln 2}.
\label{Eq:l_j}
\end{equation}
The time spent in the $j$-th round can be computed as
$\frac{|\mathbb{E}| t_r k_j^*}{\ln 2}$.
Therefore, the total execution time of Phase 1, denoted as $T_1$, can be derived as
\begin{equation}
T_1 = \sum_{j=1}^{J} \frac{|\mathbb{E}| t_r k_j^*}{\ln 2}.
\label{Eq:T1}
\end{equation}

$k_j^*$ ($1\le j\le J$), as well as $J$, are set with the parameters in Phase 2 to minimize the global execution time, as analyzed in Sec.~\ref{sec:parameter_global} and Sec.~\ref{sec:parameter_global_expected}.

Let $N^*$ be the number of tags still active after Phase 1 (i.e., $J$ rounds), it holds that
\begin{equation}
N^* = |\mathbb{E}|-m+ |\mathbb{U}_r|,
\end{equation}
where $\mathbb{U}_r$ is the set of unexpected tags still active after Phase 1.
Recall~\eqref{Eq:p_k}, the expectation of $N^*$ can be derived as
\begin{align}
E[N^*] &= |\mathbb{E}|-m+|\mathbb{U}| \prod_{j=1}^{J} {P_{1,j}} \nonumber\\
       &= |\mathbb{E}|-m+|\mathbb{U}| \big(\frac{1}{2} \big)^{\sum_{j=1}^{J} k_j^*}.
\label{Eq:E_N^*}
\end{align}

\subsection{Tuning parameters in Phase 2}

Similar to Phase 1, the false positive rate of the $w$-th round in Phase 2, defined as $P_{2,w}$, can be calculated as
\begin{equation}
P_{2,w} = \left[1-\left(1-\frac{1}{f_w}\right)^{N^* R_w}\right]^{R_w} \approx (1-e^{-{N^* {R_w}}/f_w})^{R_w}.
\label{Eq:P_2}
\end{equation}
Therefore, the Bloom filter size is
\begin{equation*}
f_w = \frac{-N^* R_w}{\ln (1- P^{\frac{1}{R_w}}_{2,w})}.
\end{equation*}

Moreover, the probability that at least one missing tag can be detected in $w$-th round, denoted as $P_{d,w}$, can be computed as
\begin{equation}
P_{d,w} = 1-{P^m_{2,w}}.
\label{Eq:P_d}
\end{equation}

Following the analysis above, the probability $P_{sys}$ that the reader is able to detect a missing event after at most $W$ rounds in Phase 2, can thus be written as
\begin{equation}
P_{sys}=1-\prod_{w=1}^{W}{(1-P_{d,w})} = 1- P^{mW}_{2,w}.
\end{equation}

It follows from the system requirement that
\begin{equation}
P_{sys}=1- P^{mW}_{2,w} = \alpha.
\label{Eq:P_system}
\end{equation}
As a result, we can obtain
\begin{equation}
f_w= \frac{-N^* R_w}{\ln (1- (1-\alpha)^{\frac{1}{m W R_w}})}.
\label{Eq:fw}
\end{equation}

In the following lemma, we derive the optimum frame size of the Bloom filter $f_w$ which is broadcast by the reader in each round of Phase 2.
\begin{lemma}
\label{Lem:Par_II}
Let $y\triangleq W R_w$, the optimum Bloom filter frame size, denoted by $f_w^*$, that achieves the detection requirement while minimising the execution time of Phase 2, is as follows:
\begin{align}
f_w^* = \frac{-N^* R_w}{\ln (1- (1-\alpha)^{\frac{1}{m y^*}})}
\label{Eq:fw_op}
\end{align}
where $y^*= \frac{\ln (1-\alpha)}{m \ln \frac{1}{2}}$.
\end{lemma}
\begin{proof}
Denote by $f$ the total length of all $W$ Bloom filters in the second phase, we thus have
\begin{equation}
f = {\sum_{w=1}^{W}{f_w}}
  = \frac{-N^* W R_w}{\ln (1- (1-\alpha)^{\frac{1}{m W R_w}})}.
\label{Eq:f_M_alpha}
\end{equation}
It can be checked that $f$ depends on the product of $W$ and $R_w$ which is the total number of hash functions used in Phase 2.
To minimize the execution time, let $y\triangleq W R_w$, we first calculate the derivation of $f$ with respect to $y$ as follows:
\begin{multline*}
\frac{\mathbf{d}f}{\mathbf{d}y}
= \frac{ N^* (1-\alpha)^{\frac{1}{m y}} \ln(1-\alpha)}{m y (1-(1-\alpha)^{\frac{1}{m y}}) \ln^2 (1-(1-\alpha)^{\frac{1}{m y}})} \\
- \frac{N^*}{\ln(1-(1-\alpha)^{\frac{1}{m y}})}.
\end{multline*}
Imposing $\frac{\mathbf{d}f}{\mathbf{d}y}=0$ yields
\begin{equation*}
y= \frac{\ln (1-\alpha)}{m \ln \frac{1}{2}}.
\end{equation*}
Moreover, when $y<\frac{\ln (1-\alpha)}{m \ln \frac{1}{2}}$, it holds that $\frac{\mathbf{d}f}{\mathbf{d}y}<0$; when $y>\frac{\ln (1-\alpha)}{m \ln \frac{1}{2}}$, it holds that $\frac{\mathbf{d}f}{\mathbf{d}y}>0$. Therefore, $f$ achieves the minimum at $y^*= \frac{\ln (1-\alpha)}{m \ln \frac{1}{2}}$. The minimum of $f_w$, denoted by $f_w^*$ can be computed by injecting $y=y^*$ into~\eqref{Eq:fw}. The proof is thus completed.
\end{proof}

\begin{remark}

As the reader does not have prior knowledge on $m$, the number of missing tags, in the design of BMTD, we require that the detection performance requirement to be hold for any $m\ge M$. Hence, $f_w^*$ and $y^*$ are as follows:
\begin{eqnarray}
f_w^* = \frac{-N^* R_w}{\ln (1- (1-\alpha)^{\frac{1}{M y^*}})} \label{Eq:f_w_op}, \\
\text{ where } y^*= \frac{\ln (1-\alpha)}{M \ln \frac{1}{2}} \label{Eq:y_op},
\end{eqnarray}
where we use $m=M$ in $N^*$ and $y^*$, which is the hardest case.
Since $N^* = |\mathbb{E}|-m+ |\mathbb{U}_r|$, it can be checked that the detection probability $P_{sys}$ is monotonically increasing and $P_{2,w}$ is monotonically decreasing with respect to the number of missing tags $m$, meaning that $m=M$ makes the detection hardest and any greater $m$ will ease the hardness, it is thus reasonable to use $m=M$ in the rest of the analysis, because if the reader can detect a missing tag event with probability $\alpha$ when $m=M$, it will fulfill the detection with probability $P_{sys}>\alpha$ when $m>M$.

In addition, since $y^*$ is the total number of hash functions used in Phase 2 and at least one round is executed so as to detect a missing event, $y^*$ needs to be a positive integer. Therefore, we set $y^*$$=$$\lceil\frac{\ln (1-\alpha)}{M \ln \frac{1}{2}}\rceil$, which guarantees the required detection performance requirement.
Note that $R_w$ and $W$ can be set as arbitrary positive integers.
\end{remark}

Under the optimum parameter setting derived above, we can calculate the time needed to execute $W$ rounds of Phase 2, denoted by $T_2$, as follows:
\begin{equation}
T_2 = \frac{-t_t N^* y^*}{\ln (1- (1-\alpha)^{\frac{1}{M y^*}})},
\end{equation}
where $t_t$ is the time needed by the tags to transmit one bit to the reader. $T_2$ sets an upper-bound on the execution time of Phase 2.


\subsection{Tuning $k_j^*$ and $J$ to minimize worst-case execution time}
\label{sec:parameter_global}

In this subsection, we study how to set $k^*_j$ and $J$ to minimize the worst-case execution time, which corresponds to the experience of the execution time where no missing event is detected and hence all the $W$ rounds in the second round need to be executed. We denote the worst-case execution time by $T$. In the following theorem, we derive the minimiser of $\mathbb{E}[T]$.


\begin{theorem}
\label{Lem:Par_I}
Denote $x \triangleq \sum_{j=1}^{J} k^*_j$, $x$ need to be set to $x^*$ as follows to minimise the worst-case execution time of the BMTD:
\begin{align}
x^* =
\begin{cases}
0 & |\mathbb{U}| \le U_0  \\
\frac{\ln \frac{-t_r |\mathbb{E}| \ln(1-(1-\alpha^{\frac{1}{My^*}}))}{t_t y^* |\mathbb{U}| \ln^2 2}}{-\ln 2} & |\mathbb{U}| > U_0
\end{cases},
\label{Eq:x*}
\end{align}
where $U_0\triangleq \frac{|\mathbb{E}| t_r \ln(1-(1-\alpha)^{\frac{1}{M y^*}})}{-t_t y^* \ln^2 2}$. That is, in regard to minimise the worst-case execution time, when the number of unexpected tags does not exceed a threshold $U_0$, Phase 1 is not executed, otherwise Phase 1 is executed with the parameters $k^*_j$ and $J$ set to $\sum_{j=1}^{J} k^*_j=x^*$.
\end{theorem}

\begin{proof}
Recall the two phases of BMTD and~\eqref{Eq:T1}, we can derive the expectation of $T$ as follows:
\begin{align}
&\mathbb{E}[T]
= T_1 + T_2
= \sum_{j=1}^{J} \frac{|\mathbb{E}| t_r k^*_j}{\ln 2} + \frac{-t_t y^* E[N^*]}{\ln (1- (1-\alpha)^{\frac{1}{M y^*}})} \nonumber\\
&= \frac{|\mathbb{E}| t_r}{\ln 2} \sum_{j=1}^{J} k^*_j
  + \frac{-t_t y^* \left(|\mathbb{E}|-M+|\mathbb{U}| \big(\frac{1}{2} \big)^{\sum_{j=1}^{J} k_j} \right)}{\ln (1- (1-\alpha)^{\frac{1}{M y^*}})}.
\label{Eq:ET}
\end{align}
From~\eqref{Eq:ET}, it can be noted that $E[T]$ is a function of $x=\sum_{j=1}^{J} k^*_j$. We then calculate the optimum $x^*$ that minimizes $E[T]$.
To that end, we compute the derivation of $E[T]$ with respect to $x$:
\begin{equation}
\frac{\mathbf{d}E[T]}{\mathbf{d}x}
= \frac{|\mathbb{E}| t_r}{\ln 2}+ \frac{t_t y^* |\mathbb{U}|\ln2}{\ln(1-(1-\alpha)^{\frac{1}{M y^*}})} \big(\frac{1}{2}\big)^{x}.
\end{equation}
Since $\big(\frac{1}{2}\big)^{x} \le 1$, it thus holds for all $x \ge 0$ that $\frac{\mathbf{d}E[T]}{\mathbf{d}x} \ge 0$ if $\frac{|\mathbb{E}| t_r}{\ln 2}+ \frac{t_t y^* |\mathbb{U}|\ln2}{\ln(1-(1-\alpha)^{\frac{1}{M y^*}})} \ge 0$, i.e.,
\begin{equation}
|\mathbb{U}| \le \frac{|\mathbb{E}| t_r \ln(1-(1-\alpha)^{\frac{1}{M y^*}})}{-t_t y^* \ln^2 2}=U_0.
\end{equation}
It is worth noticing that $E[T]$ is a monotonic nondecreasing function in this case with respect to $x$, we thus set $x=0$ to minimize the execution time, which means that if the number of unexpected tags is smaller than the threshold $U_0$, we should remove the Phase 1 and only execute Phase 2.

In contrast, if $|\mathbb{U}| > U_0$, $\frac{\mathbf{d}E[T]}{\mathbf{d}x}$ can be negative, zero, or positive.
Setting $\frac{\mathbf{d}E[T]}{\mathbf{d}x} =0$, the optimal value of $x$ to minimise $E[T]$, defined as $x^*$, can be calculated as
\begin{equation*}
x^* = \frac{\ln \frac{-t_r |\mathbb{E}| \ln(1-(1-\alpha)^{\frac{1}{My^*}})}{t_t y^* |\mathbb{U}| \ln^2 2}}{-\ln 2}.
\end{equation*}
\end{proof}

\begin{remark}
Since $x^*$ represents the total number of hash functions used in Phase 1, it needs to be a non-negative integer. Therefore, we set $x^*$ either to its ceiling or floor integer depending on which one leads to a smaller $E[T]$.
The parameters $k^*_j$ and $J$ are set such that $\sum_{j=1}^{J}k^*_j = x^*$.
\end{remark}
%

\subsection{Tuning $k_j^*$ and $J$ to minimize expected detection time}
\label{sec:parameter_global_expected}
The parameters derived in Theorem~\ref{Lem:Par_I} establish that the BMTD is able to detect a missing event with probability equal to or greater than the system requirement $\alpha$ after $W$ rounds of Phase 2. However, in many practical scenarios, the missing event may be detected in the round $w<W$ when the algorithm can be terminated. In this subsection, we derive the parameter configuration (i.e., $k^*_j$ and $J$) that minimises the expected detection time.
To that end, we first calculate the probability that at least one of the missing tags can be detected for the first time in a given slot and use it to formulate the expectation of the missing event detection time.

\begin{lemma}
\label{Lem:Prob_de}
The probability that a missing tag can be detected in a given slot of Phase 1, denoted by $q$, is as follows:
\begin{equation}
q = \left( 1-(1-(1-\alpha)^{\frac{1}{y^*M}})^\frac{M}{N^*} \right) \cdot \left(1-(1-\alpha)^{\frac{1}{y^*M}}\right).
\end{equation}
A loose lower-bound for $q$, denoted as $q_{min}$, can be established as follows:
\begin{equation}
q_{min}=\big(1-(\frac{1}{2})^\frac{M}{|\mathbb{E}|-M+|\mathbb{U}|}\big)(1-(1-\alpha)^{\frac{1}{y^*M}}).
\label{eq:q_min}
\end{equation}
\end{lemma}

\begin{proof}
A missing tag can be detected in a given slot only when at least one missing tag is hashed to this slot and no tag in $\mathbb{E}_r \cup \mathbb{U}_r$ selects the same location. Consider the hardest case for detecting a missing tag event, i.e., $m=M$, the probability that at least one missing tag maps to the given slot can be given by $\left( 1- (1-\frac{1}{f^*_w})^{M R_w} \right)$. The probability that no tag in $\mathbb{E}_r \cup \mathbb{U}_r$ maps to that slot is equal to $(1-\frac{1}{f^*_w})^{N^* R_w}$. Consequently, multiplying the former by the later leads to $q$, i.e.:
\begin{align*}
q &= \left( 1- (1-\frac{1}{f^*_w})^{M R_w} \right)\cdot(1-\frac{1}{f^*_w})^{N^* R_w} \\
  &\approx (1-e^{-\frac{MR_w}{f^*_w}})\cdot e^{-\frac{N^*R_w}{f^*_w}} \\
  &= \left( 1-(1-(1-\alpha)^{\frac{1}{y^*M}})^\frac{M}{N^*} \right) \cdot (1-(1-\alpha)^{\frac{1}{y^*M}}).
\end{align*}
We then derive the lower-bound $q_{min}$. To that end, noticing that $q$ is negatively correlated with $N^*$ which falls into the range $\big[|\mathbb{E}|-M,|\mathbb{E}|-M+|\mathbb{U}|\big]$, we have
$$q\ge \left( 1-(1-(1-\alpha)^{\frac{1}{y^*M}})^\frac{M}{|\mathbb{E}|-M+|\mathbb{U}|} \right) \cdot (1-(1-\alpha)^{\frac{1}{y^*M}}).$$
On the other hand, noticing that $y^*=\lceil{\frac{\ln (1-\alpha)}{M\ln \frac{1}{2}}\rceil}\ge \frac{\ln (1-\alpha)}{M\ln \frac{1}{2}}$, we have $q\ge q_{min}=\big(1-(\frac{1}{2})^\frac{M}{|\mathbb{E}|-M+|\mathbb{U}|}\big)(1-(1-\alpha)^{\frac{1}{y^*M}})$.
\end{proof}

After calculating $q$, we next derive the expected missing event detection time, denoted by $\mathbb{E}[T_D]$.

\begin{theorem}
\label{Th:T_Z}
The expected missing event detection time $\mathbb{E}[T_D]$ is given by the following equation:
\begin{align}
\label{Eq:E[T_D]}
&\mathbb{E}[T_D]
= \frac{|\mathbb{E}| t_r x}{\ln 2} + t_t \sum_{N^*=|\mathbb{E}|-M}^{|\mathbb{E}|-M+|\mathbb{U}|}
    \frac{1-(1-q)^{f}-fq(1-q)^f}{q} \nonumber\\
         & \binom{|\mathbb{U}|}{N^*-|\mathbb{E}|+M}
          \Big(\frac{1}{2^{x}}\Big)^{N^*-|\mathbb{E}|+M} \Big(1-\frac{1}{2^{x}}\Big)^{|\mathbb{U}|-N^*+|\mathbb{E}|-M}.
\end{align}
\end{theorem}

\begin{proof}
Recall~\eqref{Eq:f_M_alpha}, it holds that there are  $f=\frac{-N^* y^*}{\ln (1- (1-\alpha)^{\frac{1}{M y^*}})}$ slots in Phase 2. We next calculate the number of slots before detecting the first missing tag.
It is easy to check that the event that in slot $z$ the reader detects the first missing tag happens if no missing tags is detected in the first $z-1$ slots while at least one missing tag is detected in slot $z$. Let $Z$ denote the random variable of $z$, we have
\begin{equation}
P\{Z=z\} = (1-q)^{z-1}*q,
\label{eq:z}
\end{equation}
which is geometrically distributed.

We can then compute the expectation of $Z$, conditioned by $N^*$, as follows:
\begin{align}
E[Z|N^*] &= \sum_{z=1}^{f} z \cdot P\{Z=z\}  \nonumber \\
&= \frac{1-(1-q)^{f}-fq(1-q)^f}{q}.
\label{eq:e_zn}
\end{align}

Moreover, it follows from the analysis of Phase 1 that the probability that an unexpected tag is still active after Phase 1 is $\prod_{j=1}^{J} {P_{1,j}}$.
On the other hand, since $\mathbb{U}_r$ represents the ID set of active unknown tags after Phase 1, recall~\eqref{Eq:p_k} and $\sum_{j=1}^{J} k^*_j =x$, we can compute the probability of having $u$ active unexpected tags after Phase 1 as follows:
\begin{align}
P\{|\mathbb{U}_r| = u\}
&= \binom{|\mathbb{U}|}{u} \Big(\prod_{j=1}^{J} {P_{1,j}}\Big)^u \Big(1-\prod_{j=1}^{J}
     {P_{1,j}}\Big)^{|\mathbb{U}|-u} \nonumber\\
&= \binom{|\mathbb{U}|}{u} \Big(\frac{1}{2^{x}}\Big)^{u} \Big(1-\frac{1}{2^{x}}\Big)^{|\mathbb{U}|-u}.
\end{align}
It can be noted that $|\mathbb{U}_r|$ follows the binomial distribution.
Recall the relationship between $N^*$ and $|\mathbb{U}_r|$ in~\eqref{Eq:T1}, it holds that
\begin{align}
E[Z] =& \sum_{N^*=|\mathbb{E}|-M}^{|\mathbb{E}|-M+|\mathbb{U}|} E[Z|N^*]
         \binom{|\mathbb{U}|}{N^*-|\mathbb{E}|+M} \cdot \nonumber\\
      &    \Big(\frac{1}{2^{x}}\Big)^{N^*-|\mathbb{E}|+M}
    \cdot \Big(1-\frac{1}{2^{x}}\Big)^{|\mathbb{U}|-N^*+|\mathbb{E}|-M}
    \label{eq:e_z}
\end{align}

Therefore, $E[T_D]$ can be derived as
\begin{align}
E[T_D]
= T_1 + E[Z] \cdot t_t
= \frac{|\mathbb{E}| t_r x}{\ln 2} + E[Z] \cdot t_t.
\label{eq:t_d}
\end{align}
Injecting $E[Z]$ into $E[T_D]$ completes the proof.
\end{proof}



After deriving $E[T_D]$ as a function of $x$, we seek the optimum, denoted by $x^*_e$, which minimizes $E[T_D]$. To this end, we first establish an upper-bound of $x^*_e$ in the following lemma.

\begin{lemma}
\label{Lem:x_opt_upp}
It holds that $x^*_e \le \frac{2t_t \ln2}{t_r |\mathbb{E}| q_{min}}$.
\end{lemma}

\begin{proof}
We write $E[T_D]$ as a function of $x$. Specifically, let $E[T_D]=g(x)$. To prove the lemma, we show that for any $x>2x_0$ it holds that $g(x)\ge g(x_0)$ where $x_0\triangleq \frac{t_t \ln2}{t_r |\mathbb{E}| q_{min}}$.

To this end, we first derive the bounds of $g(x)$.
Recall~\eqref{eq:z},\eqref{eq:e_zn},~\eqref{eq:e_z} and~\eqref{eq:t_d}, we have
\begin{eqnarray*}
g(x)&>&\frac{|\mathbb{E}| t_r x}{\ln 2},  \\
g(x)&\le& \frac{|\mathbb{E}| t_r x}{\ln 2} + \frac{t_t}{q_{min}}.
\end{eqnarray*}

For any $x>2x_0$, we then have
$$g(x)>\frac{|\mathbb{E}| t_r x}{\ln 2}>\frac{2|\mathbb{E}| t_r x_0}{\ln 2}=\frac{|\mathbb{E}| t_r x_0}{\ln 2} + \frac{t_t}{q_{min}}\ge g(x_0)$$
The lemma is thus proved.
\end{proof}

Lemma~\ref{Lem:x_opt_upp} shows that $x^*_e$ falls into the range $[0,2x_0]$. We can thus search $[0,2x_0]$ to find $x^*_e$ that minimises $E[T_D]$ and then set $J$ and $k_j^*$ such that $\sum_{j=1}^{J}k_j^* = x^*_e$.

\subsection{BMTD parameter setting: summary}

We conclude this section by streamlining the procedure of the parameter setting in the BMTD:
\begin{enumerate}
\item \emph{Set parameters in Phase 2:} given ${|\mathbb{E}|}$, $M$, $\alpha$ and ${|\mathbb{U}|}$, compute $f^*_w$ and $y^*$
      by~\eqref{Eq:f_w_op} and~\eqref{Eq:y_op}, respectively, and set $R_w$ and $W$ such that $R_wW=y^*$;
\item \emph{Set parameters in Phase 1:} compute $x^*$ by Theorem~\ref{Lem:Par_I} if the objective is to minimise the worst-case execution time; compute $x^*_e$ if the objective is to minimise the expected detection time; then the set of $k_j^*$ and $J$ is given such that $\sum_{j=1}^J k_j^*=x^*$ or $\sum_{j=1}^J k_j^*=x^*_e$.
\end{enumerate}
Following the above two steps, we can obtain all parameters in the BMTD.

\section{Cardinality estimation}
\label{sec:estimation}
In order to execute the BMTD, the reader needs to estimate the number of unexpected tags $|\mathbb{U}|$. In our work, we use the SRC estimator which is designed in~\cite{chen2013understanding} and is the current state-of-the-art solution. Denote by $\overline{|\mathbb{E}|-m + |\mathbb{U}|}$ the estimated total number of tags in the system, then the cardinality $|\mathbb{U}|$ can be approximated as $\overline {|\mathbb{U}|} = \overline{|\mathbb{E}|-m + |\mathbb{U}|}-|\mathbb{E}|$ if $m<<|\mathbb{E}|, |\mathbb{U}|$.
Because the number of bits that set to one in Bloom filter is concentrated tightly around the mean~\cite{mitzenmacher2005probability} and~\cite{hao2007building}, once the estimation $\overline {|\mathbb{U}|}$ is obtained, we can calculate the expectation of $N^*$ according to~\eqref{Eq:E_N^*} with $m=M$ and use it as the estimator of $N^*$.

The SRC estimator consists of two phases: rough estimation and accurate estimation. It is proven in~\cite{chen2013understanding} that SRC can obtain a rough estimation $\hat{n}$ which at least equals to $0.5({|\mathbb{E}|-m + |\mathbb{U}|}) $ after its first phase. In the second phase, SRC can achieve that the relative estimation error is not greater than $\epsilon$ which is referred to as confidence range with the settings as follows: the frame size $L_{est}=\frac{65}{(1-0.04^\epsilon)^2}$ and the persistence probability $p_{pe} = \min\{1,1.6 L_{est}/\hat{n}\}$.

We then analyse the overhead introduced to estimate the cardinality of $\mathbb{U}$. As proven in~\cite{chen2013understanding}, the overhead of SRC estimator is at most $O(\frac{1}{\epsilon^2}+\log\log(|\mathbb{U}|+|\mathbb{E}|))$, which is moderate for large-scale RFID systems with large $|\mathbb{U}|$ and $|\mathbb{E}|$.

\subsection{Fast detection of missing event}

In our estimation approach, we require that $m\ll |\mathbb{E}|, |\mathbb{U}|$. In case where $m$ is close to $|\mathbb{E}|, |\mathbb{U}|$, the estimation may not be accurate. Luckily, in this case, we can quickly detect a missing event in the cardinality estimation phase due to large $m$.

Specifically, we analyze the SRC estimator's capability of detecting missing event under large $m$ by comparing the pre-computed slots with those selected by the present tags. Recall the proof of Lemma~\ref{Lem:Prob_de}, we can derive the detection probability in any given slot, defined as $q_{pre}$, as
\begin{equation}
q_{pre} = \left(1-\big(1-\frac{p_{pe}}{L_{est}}\big)^{m}\right)*
            \Big(1-\frac{p_{pe}}{L_{est}}\Big)^{(\mathbb{U}+\mathbb{E}-m)}.
\end{equation}
Since the detections in different slots are independent of each other, the probability of detecting at least one missing tag event by the SRC estimator can be calculated as $1-(1-q_{pre})^{L_{est}}$ which is a increasing function of $m$.

\begin{figure}[htbp]
\centering
\includegraphics[width=7cm]{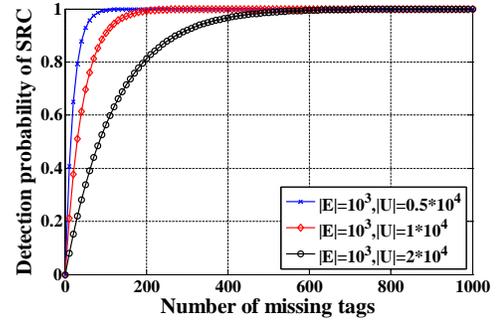}
\caption{$q_{pre}$ vs. $m$.}
\label{Fig:SRC}
\end{figure}

Fig.~\ref{Fig:SRC} illustrates the detection probability of SRC with the various number of missing tags under different unexpected tag population sizes. To obtain the figure, we set $|\mathbb{E}|=10^3$ and $\epsilon=0.1$. It is observed that in the cases that $|\mathbb{U}|=0.5*10^4$, $1*10^4$, $2*10^4$, SRC is able to detect at least a missing tag event with probability one when $m$ is not less than $100$, $200$, $600$, which means that a missing event is detected by SRC and the reader does not need to invoke the BMTD. In the other side, in the cases that $m$ is less than $100$, $200$, $600$, it holds that $|{\frac{\overline{|\mathbb{U}|}}{|\mathbb{U}|}-1}|$$\le$ $0.138,0.132,0.128$, respectively. With reference to the conclusion drawn from the Fig.~\ref{Fig:error_estimation}, the BMTD can tolerate these levels of estimation error.

\subsection{Sensibility to estimation error}

The estimation algorithm we use inevitably introduces error on $|\mathbb{U}|$, which may have a negative impact on the performance of the BMTD. In order to investigate this impact, we next illustrate the sensitivity of the detection time to the estimation error.

\begin{figure}[htbp]
\centering
\centering
\includegraphics[width=7cm]{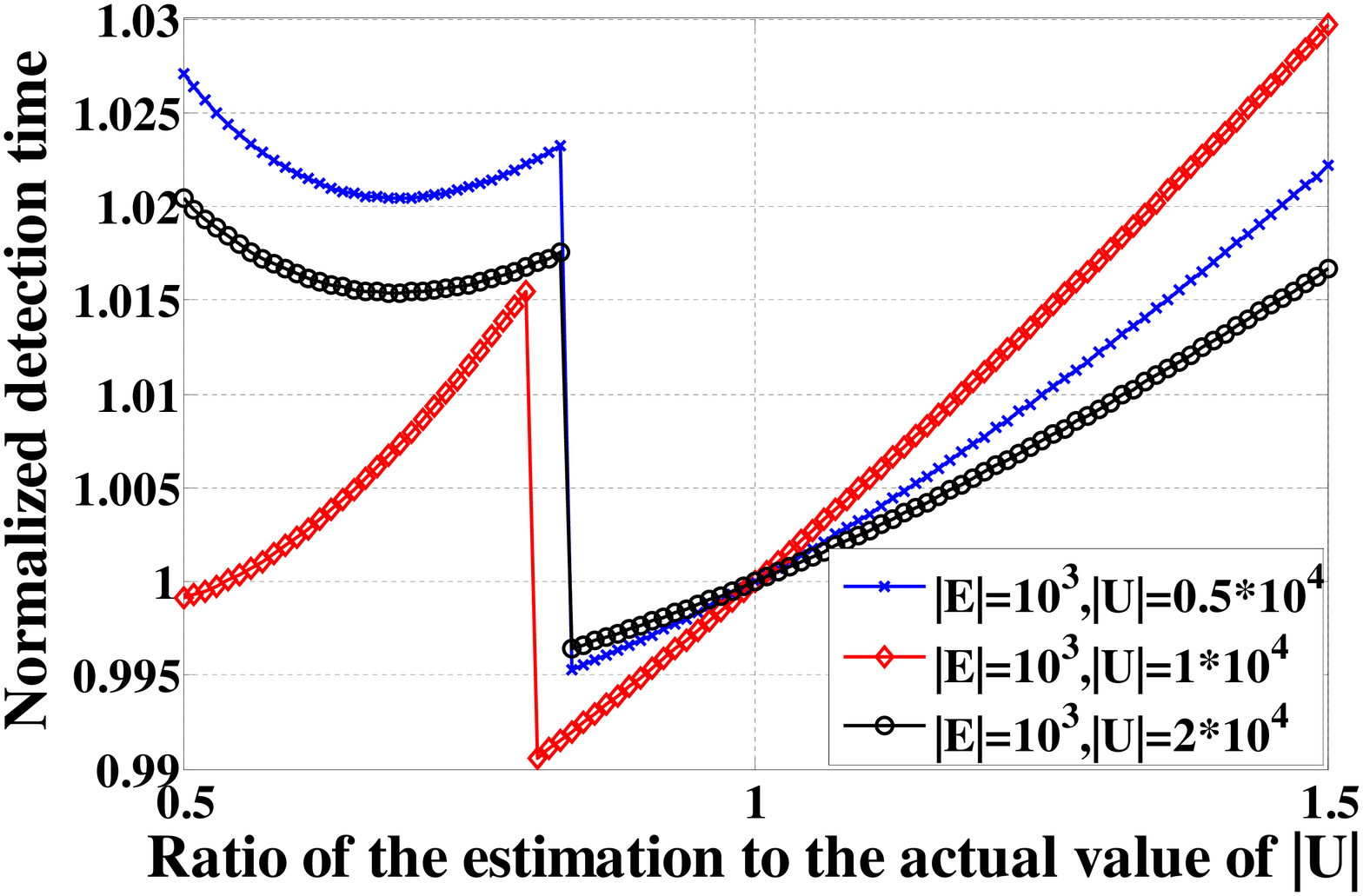}
\caption{$\frac{\overline {E[{T_D}]}}{E[T_D]}$ vs. $\frac{\overline{|\mathbb{U}|}}{|\mathbb{U}|}$.}
\label{Fig:error_estimation}
\end{figure}

Fig.~\ref{Fig:error_estimation} shows the theoretically calculated expected detection time from~\eqref{Eq:E[T_D]} under different unexpected tag population sizes and various levels of estimation error for $M=1$. All results here are normalized with respect to the expected detection time without estimation error, which can be represented as $\frac{\overline {E[{T_D}]}}{E[T_D]}$. As shown in the figure, the expected detection time based on the estimation is greater than the actual one ${E[T_D]}$ almost in all levels of estimation error. But it is worth noticing that the expected detection time only increases by up to $0.5\%$ when $|{\frac{\overline{|\mathbb{U}|}}{|\mathbb{U}|}-1}| \le 0.1$, which is nearly same with that without estimation error. Even when $|{\frac{\overline{|\mathbb{U}|}}{|\mathbb{U}|}-1}| = 0.5$, the departure from the detection time without estimation is only $3\%$. Therefore, it can be concluded that BMTD is very robust to the estimation error.

\subsection{Enforcing detection reliability}

Estimation error also has impact on the reliability of the BMTD as $P_{sys}$ is calculated base on the estimated cardinality.


To enforce the detection reliability, we introduce more rounds to execute additional Bloom filters. The scheme works as follows: After receiving the Bloom filtering vector constructed by the active tags in the set $\mathbb{E}_r \cup \mathbb{U}_{r}$ in each round of Phase 2, the reader first counts the actual number of '1' bits in the filtering vector, defined as $s_1$ and uses it to compute the actual false positive probability, denoted by $\hat{P}_{2,w}$, as follows:
\begin{equation}
\hat{P}_{2,w} = \frac{s_1}{f^*_w},
\end{equation}
because an arbitrary unexpected tag maps to a '1' bit with a probability of $s_1$ out of $f^*_w$.

Following~\eqref{Eq:P_system}, we have the observed protocol reliability, denoted by $\hat P_{sys}$, as follows:
\begin{equation}
\hat P_{sys}=1- \hat P^{MW}_{2,w}.
\end{equation}
If $\hat P_{sys}<\alpha$, the reader adds one more round in Phase 2 to further detect the missing tag event until $\hat P_{sys} \ge \alpha$.

\subsection{Discussion on multi-reader case}


In large-scale RFID systems deployed in a large area, multiple readers are thus deployed to ensure the full coverage for a larger number of tags in the interrogation region. 
In such scenarios, we leverage the approach proposed in~\cite{kodialam2007anonymous} and employed in~\cite{shahzad2015expecting}.
The main idea is that a back-end server is used to synchronize all readers such that the RFID system with multiple readers operates as the single-reader case.

Specially, the back-end server calculates all the parameters involved in BMTD and constructs Bloom filter and sends them to all readers such that they broadcast the same parameters and Bloom filter to the tags. Furthermore,
each reader sends its individual Bloom filtering vector back to the back-end server. When the back-end server receives all Bloom filtering vectors, it applies logical $OR$ operator on all received Bloom filtering vectors, which eliminates the impact of the duplicate readings of tags in the overlapped interrogation region. Consequently, a virtual Bloom filter is constructed by the back-end server.

\section{Performance Evaluation}
\label{sec:simulation}

The problem addressed in this paper is to detect the missing expected tags in the presence of a large number of unexpected tags in a time-efficient and reliable way. In this section, we evaluate the performance of the proposed BMTD. It has been shown in~\cite{shahzad2015expecting} that existing missing detection protocols cannot achieve the required reliability when there are unexpected tags in the RFID systems except the latest RUN~\cite{shahzad2015expecting}. We thus compare our proposed BMTD to RUN in terms of the actual reliability and the detection time. Note that the detection time can be interpreted as the time taken to either detect the fist missing tag event if a missing tag is found or complete the execution if no missing tag is found.

The simulation parameters are set with reference to~\cite{luo2014missing} and~\cite{shahzad2015expecting}. Specifically, since both transmission rates from the tags to the reader and the reader to the tags depend on physical implementation and interrogation environment, we make the same assumption as in~\cite{luo2014missing} that $t_r=t_t$. Moreover, because RUN is the baseline protocol, we use the same performance metrics as in~\cite{shahzad2015expecting} where the time needed to detect a missing tag event is shown in terms of the number of slots. To that end, we, without loss of generality, assume $t_r=t_t=1$ in ~\eqref{Eq:E[T_D]} in the simulation.
Besides, we compute the optimal parameter values for RUN by following its specifications.

In the simulation, we use SRC~\cite{chen2013understanding} armed with missing tag detection function in this paper to estimate the unexpected tag population size with the confidence rang $\epsilon=0.1$. And all presented results are obtained by taking the average value of 100 independent trials under the same simulation setting.

We start by evaluating the performance of the BMTD by optimizing the worst-case execution time and the expected detection time.

\subsection{Comparison between two strategies of BMTD}

In this subsection, we compare the performance of two strategies of the BMTD which are abbreviated to Worst-$M$ and Expected-$M$ here, respectively. We set $|\mathbb{E}|=1000$, $m=100$, $\alpha=0.9$, $|\mathbb{U}|=10000:5000:30000$, $M=1$ and $50$.

Table~\ref{Tab:worst_exp} lists the results where the first and second elements in the two-tuple $(\cdot,\cdot)$ denote the actual reliability and detection time, respectively. It can be seen that Expected-$M$ costs less time than Worst-$M$ to achieve the same reliability which is greater than the system requirement on the detection reliability, especially when $M$ is small. Specifically, compared with Worst-$1$, Expected-$1$ reduces the detection time by up to $51.92\%$ when $|\mathbb{U}|=10000$. This is because $x^*=5$ is  too large for Phase 1 by optimizing the worst-case execution time, which wastes time. In contrast, minimizing the expected detection time relieves the influence of unexpected tag population size on the time of Phase 2 and thus outputs a smaller $x^*_e=2$.
In the rest of our simulation, we configure the parameters of the BMTD to minimise the expected detection time. 

\begin{table}[!htbp]
\centering
\caption{Actual reliability and detection time of BMTD}
\label{Tab:worst_exp}
\begin{tabular}{|l|l|l|l|l|l|}
\hline
\multirow{2}{*}{\textbf{Strategy}} & \multicolumn{5}{c|}{\textbf{Number of unexpected tags}}\\
\cline{2-6} & 10000 & 15000 &20000&25000&30000\\
\hline
Worst-1& (1,4108)  & (1,4441) & (1,5013) &(1,5453) &(1,5510) \\
\hline
Expected-1& (1,1975) &  (1,3187)  &  (1,3569) &  (1,3828)& (1,4191) \\
\hline
Worst-50& (1,1357)  & (1,1841) & (1,2753) & (1,2762)& (1,2995)\\
\hline
Expected-50& (1,1353) & (1,1618) &  (1,2272) & (1,2472) &(1,2815)  \\
\hline
\end{tabular}
\end{table}

\subsection{Comparison between BMTD and RUN}

\subsubsection{Comparison under different number of missing tags}

In this subsection, we evaluate the performance of BMTD under different number of missing tags, which stands for the effectiveness and efficiency of BMTD. To that end, we set $|\mathbb{E}|=1000$, $|\mathbb{U}|=30000$, $m=1:50:901$, $\alpha=0.9$ and $0.99$. Moreover, we set the threshold to $M=1$.

\textbf{Actual reliability:}
BMTD achieves the required reliability for any missing tag population size when there are a large number of unexpected tags in the RFID systems.
Fig.~\ref{Fig:R_09_m} and~\ref{Fig:R_099_m} illustrate the actual reliability of BMTD and RUN for $\alpha=0.9$ and $0.99$, respectively. It can be observed that both BMTD and RUN achieve the reliability more than that required by the system.

\begin{figure}[htbp]
\centering
\subfigure[$\alpha=0.9$]{
\begin{minipage}[t]{0.46\linewidth}
\centering
\includegraphics[width=1\textwidth]{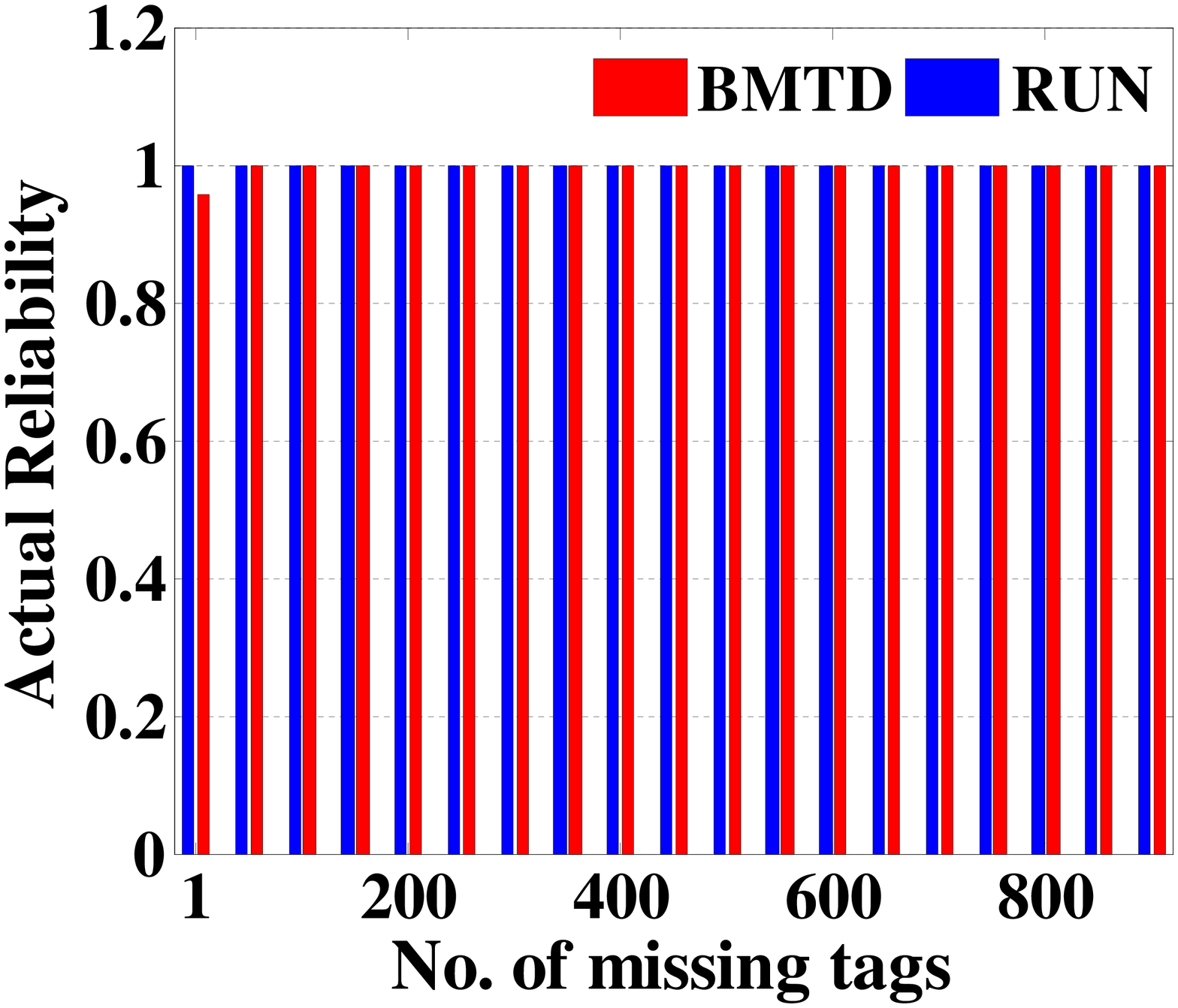}
\label{Fig:R_09_m}
\end{minipage}}
\subfigure[$\alpha=0.99$]{
\begin{minipage}[t]{0.46\linewidth}
\centering
\includegraphics[width=1\textwidth]{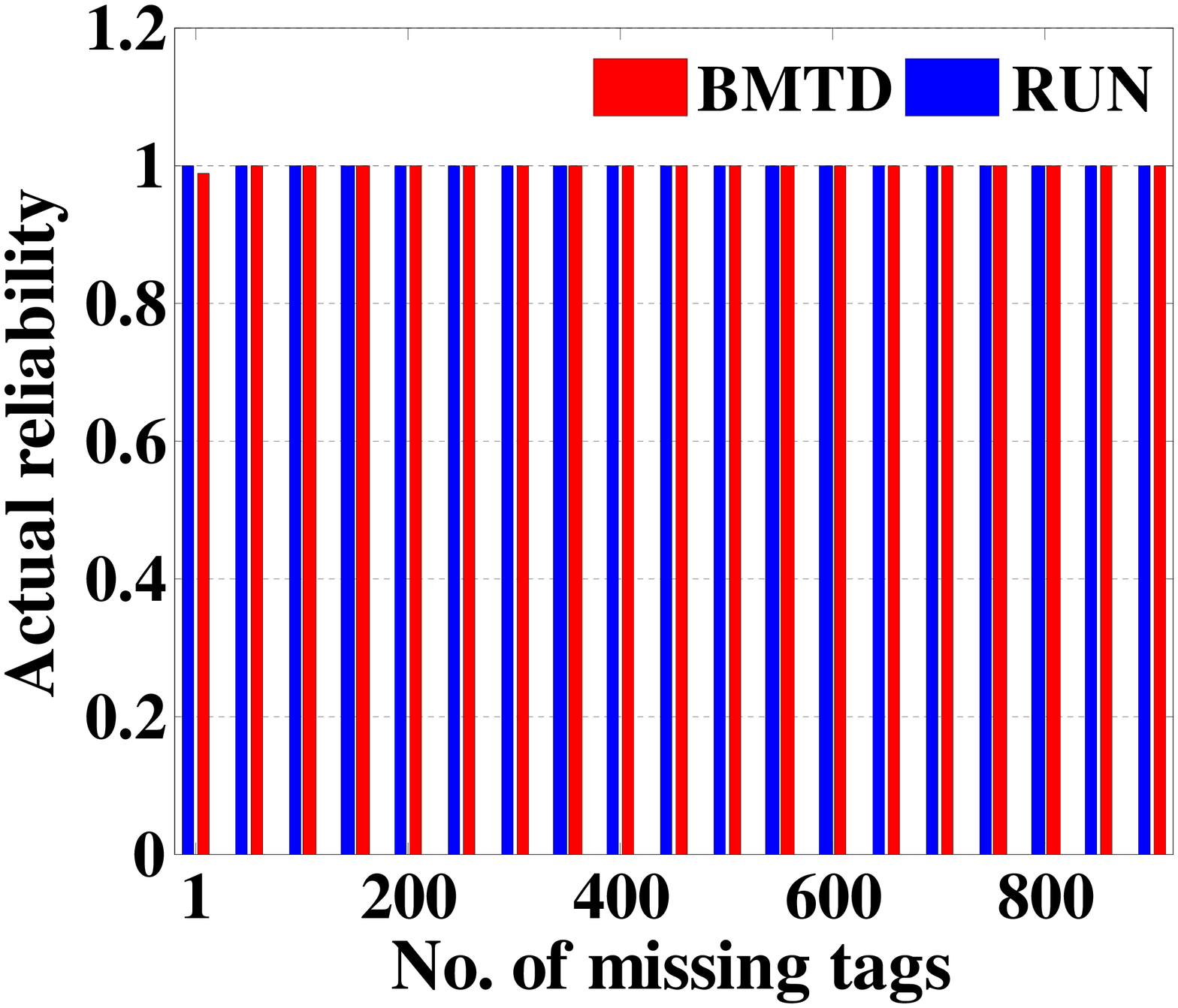}
\label{Fig:R_099_m}
\end{minipage}}
\caption{Actual reliability vs. number of missing tags}
\label{Fig:R_m}
\end{figure}

\textbf{Detection time:}
BMTD is more time-efficient in comparison to RUN.
Fig.~\ref{Fig:T_09_m} and~\ref{Fig:T_099_m} show the detection time for $\alpha=0.9$ and $0.99$, respectively. For clearness, we further highlight the caves from $m=51$ to $901$.
As shown in the figures, the detection time of BMTD is far shorter than that of RUN and decreases with the number of missing tags significantly. This is unsurprising. BMTD is able to deactivate major unexpected tags, which greatly reduces the number of active tags in the population, such that the presence of more missing tags makes the detection much easier. In contrast, RUN does not take into account the impact of unexpected tag population size, leading to longer detection delay in the presence of large number of unexpected tags.


\begin{figure}[htbp]
\centering
\subfigure[$\alpha=0.9$]{
\begin{minipage}[t]{0.46\linewidth}
\centering
\includegraphics[width=1\textwidth]{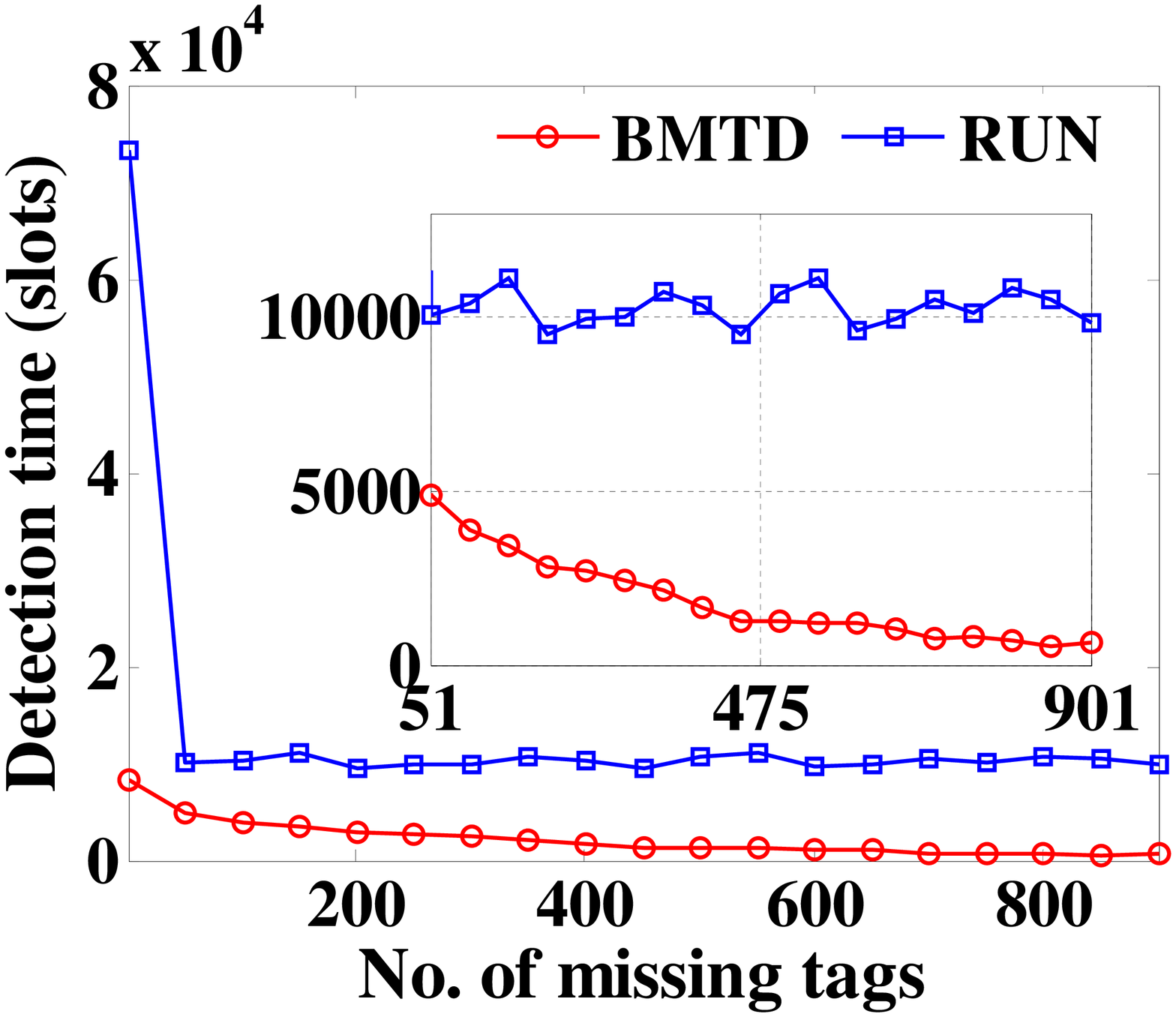}
\label{Fig:T_09_m}
\end{minipage}}
\subfigure[$\alpha=0.99$]{
\begin{minipage}[t]{0.46\linewidth}
\centering
\includegraphics[width=1\textwidth]{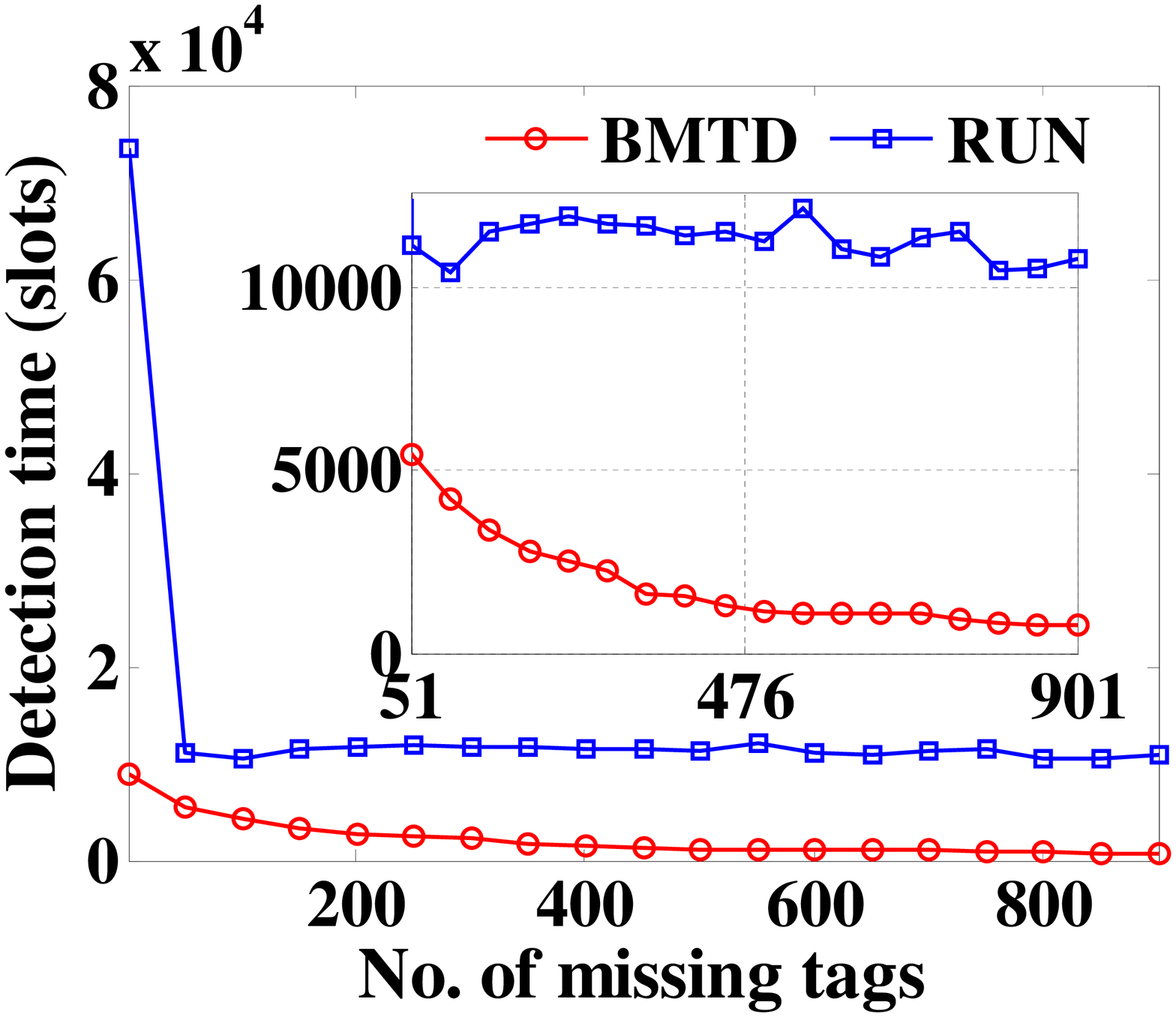}
\label{Fig:T_099_m}
\end{minipage}}
\caption{Detection time vs. number of missing tags}
\label{Fig:T_m}
\end{figure}

\subsubsection{Comparison under different number of unexpected tags}
In this subsection, we evaluate the performance of BMTD under different number of unexpected tags, which represents the generality of BMTD. To that end, we set $|\mathbb{E}|=1000$, $m=50$, $M=1$, $\alpha=0.9$ and $0.99$. Moreover, we select such $|\mathbb{U}|=1000, 5000:5000:30000$ that various values of $\frac{|\mathbb{U}|}{|\mathbb{E}|}$ are covered in the simulation.

\textbf{Actual reliability:}
BMTD achieves the reliability greater than the required reliability for different cardinalities of unexpected tag set.
Fig.~\ref{Fig:R_09_U} and~\ref{Fig:R_099_U} depict the actual reliability of BMTD and RUN for $\alpha=0.9$ and $0.99$, respectively. It can be observed that the actual reliability achieved by both BMTD and RUN is equal to one.

\begin{figure}[htbp]
\centering
\subfigure[$\alpha=0.9$]{
\begin{minipage}[t]{0.46\linewidth}
\centering
\includegraphics[width=1\textwidth]{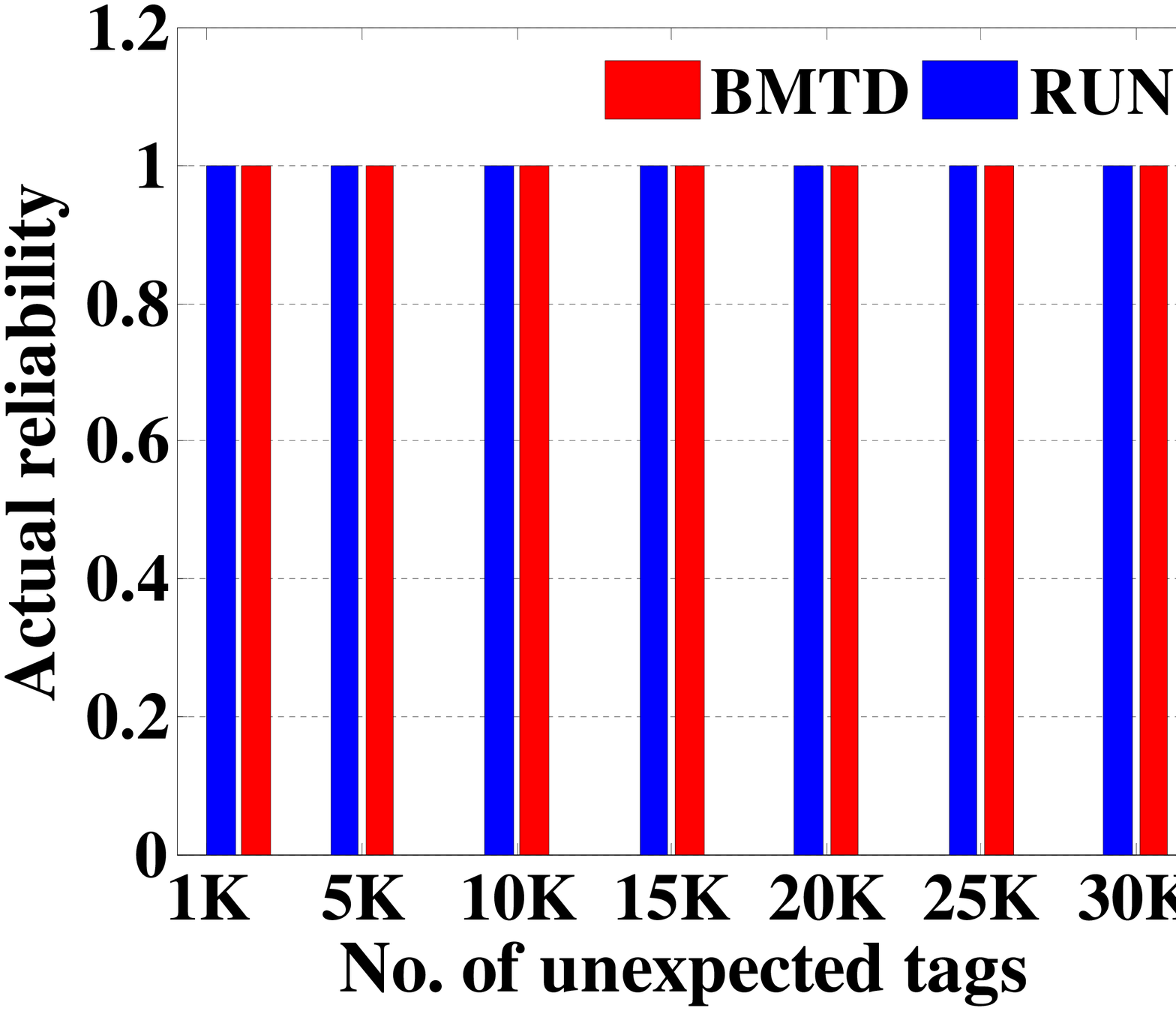}
\label{Fig:R_09_U}
\end{minipage}}
\subfigure[$\alpha=0.99$]{
\begin{minipage}[t]{0.46\linewidth}
\centering
\includegraphics[width=1\textwidth]{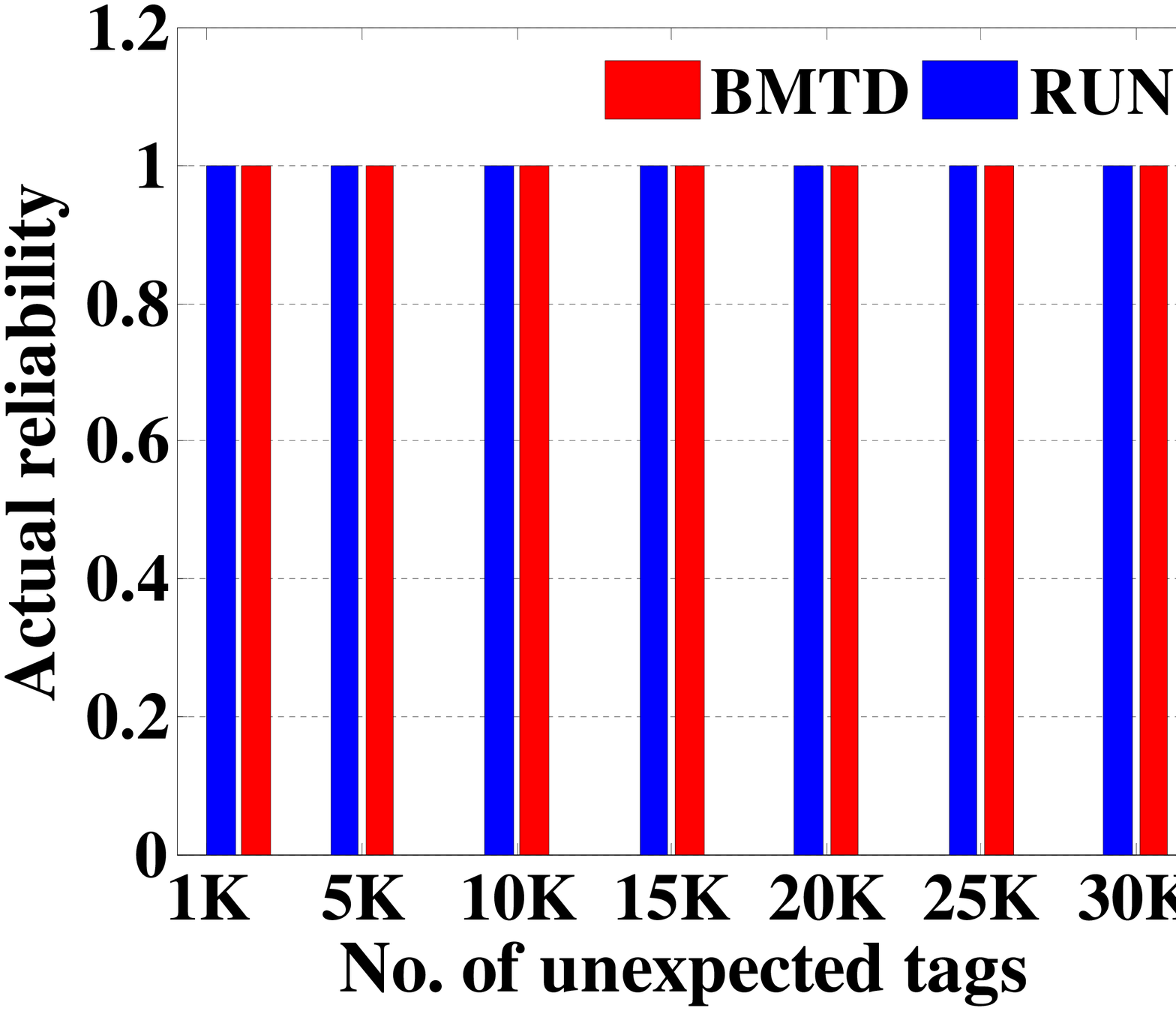}
\label{Fig:R_099_U}
\end{minipage}}
\caption{Actual reliability vs. number of unexpected tags}
\label{Fig:R_U}
\end{figure}

\textbf{Detection time:}
The BMTD outperforms the RUN considerably in terms of detection time even in the scenario with the small number of unexpected tag.
Fig.~\ref{Fig:T_09_U} and~\ref{Fig:T_099_U} show the detection time for $\alpha=0.9$ and $0.99$, respectively. As shown in the figures, BTMD is able to save time especially when more unexpected tags are present in the population. Moreover, the increase in detection time of BTMD is more slow than that of RUN. This is due to the ability of BTMD that it can detect the missing tag event when estimating the $|\mathbb{U}|$ and determine whether to execute the unexpected tag deactivation phase following Lemma~\ref{Lem:x_opt_upp}, which is exactly ignored in RUN.


\begin{figure}[htbp]
\centering
\subfigure[$\alpha=0.9$]{
\begin{minipage}[t]{0.45\linewidth}
\centering
\includegraphics[width=1\textwidth]{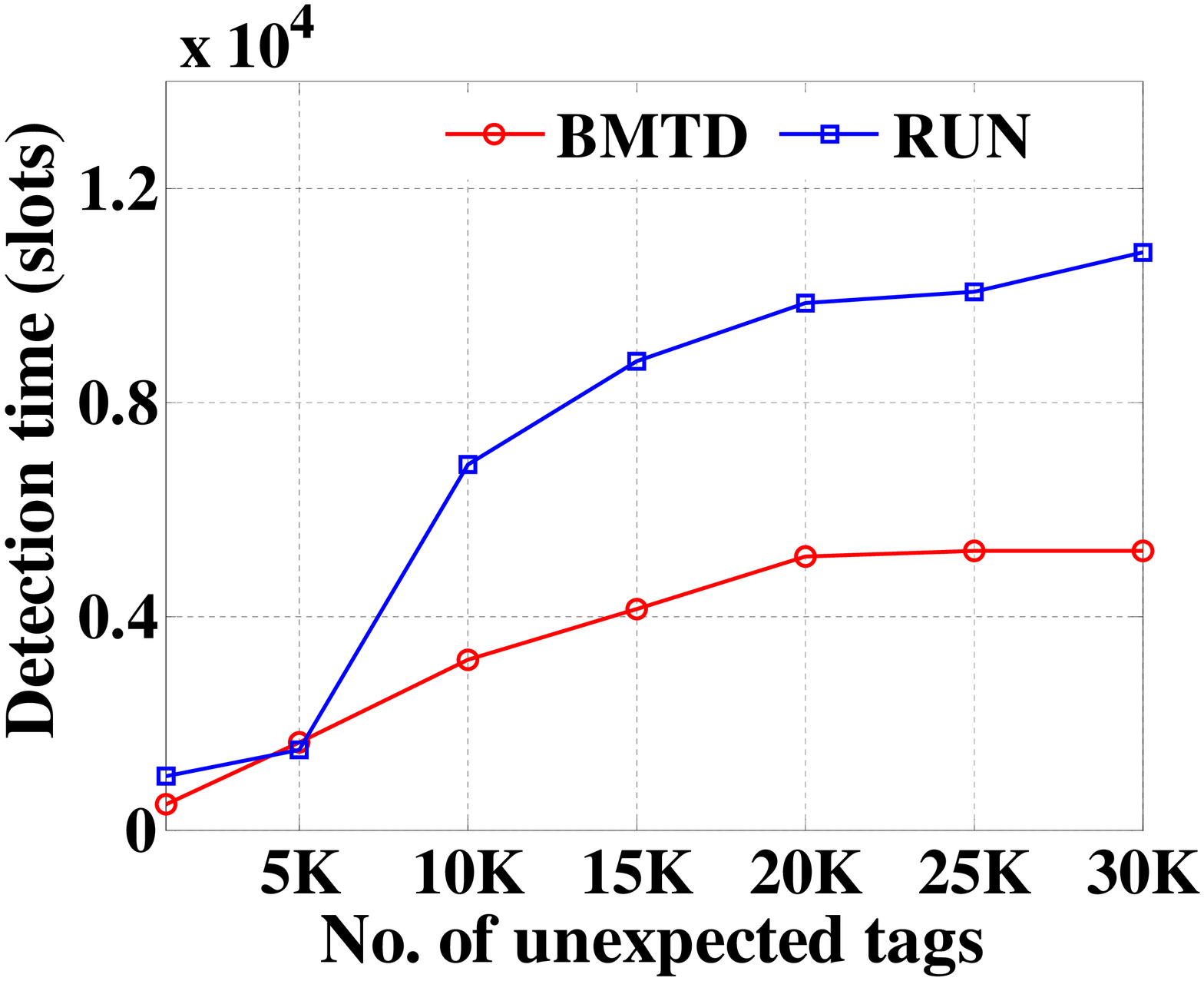}
\label{Fig:T_09_U}
\end{minipage}}
\subfigure[$\alpha=0.99$]{
\begin{minipage}[t]{0.45\linewidth}
\centering
\includegraphics[width=1\textwidth]{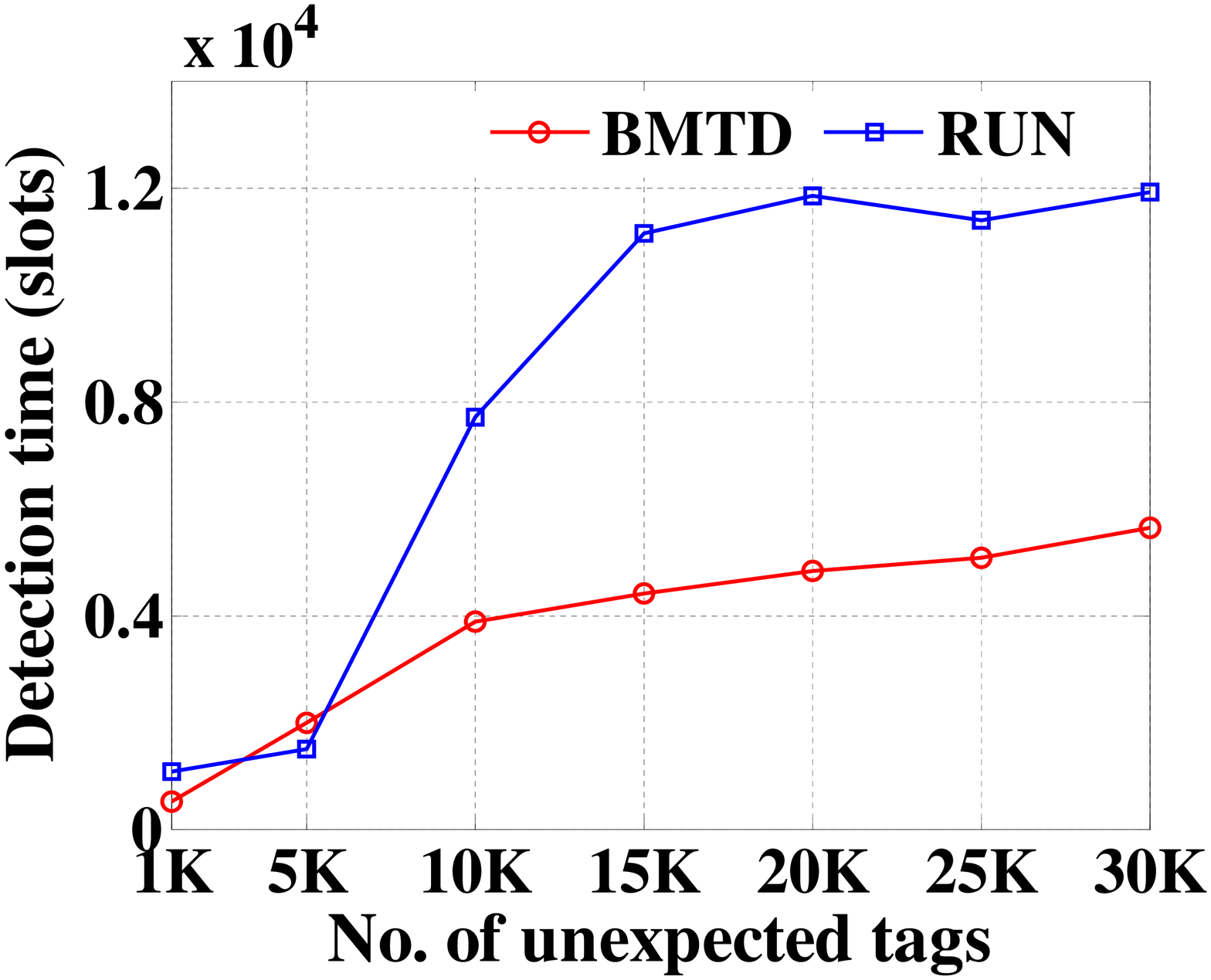}
\label{Fig:T_099_U}
\end{minipage}}
\caption{Detection time vs. number of unexpected tags}
\label{Fig:T_U}
\end{figure}

\subsubsection{Comparison under different values of threshold}

In this subsection, we evaluate the performance of BMTD under different thresholds, which represents the tolerability of BMTD. To that end, we set $|\mathbb{E}|=1000$, $|\mathbb{U}|=30000$, $m=100$, $\alpha=0.9$ and $0.99$. Moreover, we choose such $M=50:50:300$ that the threshold can be greater or smaller than or equal to the number of missing tags in the simulation.

\textbf{Actual reliability:}
BMTD achieves better reliability than the required reliability when $m \ge M$.
As shown in Fig.~\ref{Fig:R_09_M} and~\ref{Fig:R_099_M}, BMTD fails to achieve the required reliability only when $m<M$, which does not have negative impact because the objective of the missing tag detection protocol is to detect the missing tags only if the number of missing tags exceeds the threshold $M$.

\textbf{Detection time:}
BMTD can tolerate the deviation from the threshold in terms of the detection time even when $m<M$.
Fig.~\ref{Fig:T_09_M} and~\ref{Fig:T_099_M} show the detection time for $\alpha=0.9$ and $0.99$, respectively. It can be seen from the figures that the detection time of BMTD almost does not vary with the deviation. The detection time of RUN, by contrast, increases substantially as the deviation increase when $m<M$. This is because RUN terminates only when it runs optimal number of frames since the first frame when the estimated value of $|\mathbb{U}|$ does not vary by $0.1\%$ in consecutive $10$ frames if it does not detect any missing tag in any frame, while BMTD stops once the observed reliability $\hat P_{sys}$ exceeds $\alpha$.

\begin{figure}[htbp]
\centering
\subfigure[$\alpha=0.9$]{
\begin{minipage}[t]{0.46\linewidth}
\centering
\includegraphics[width=1\textwidth]{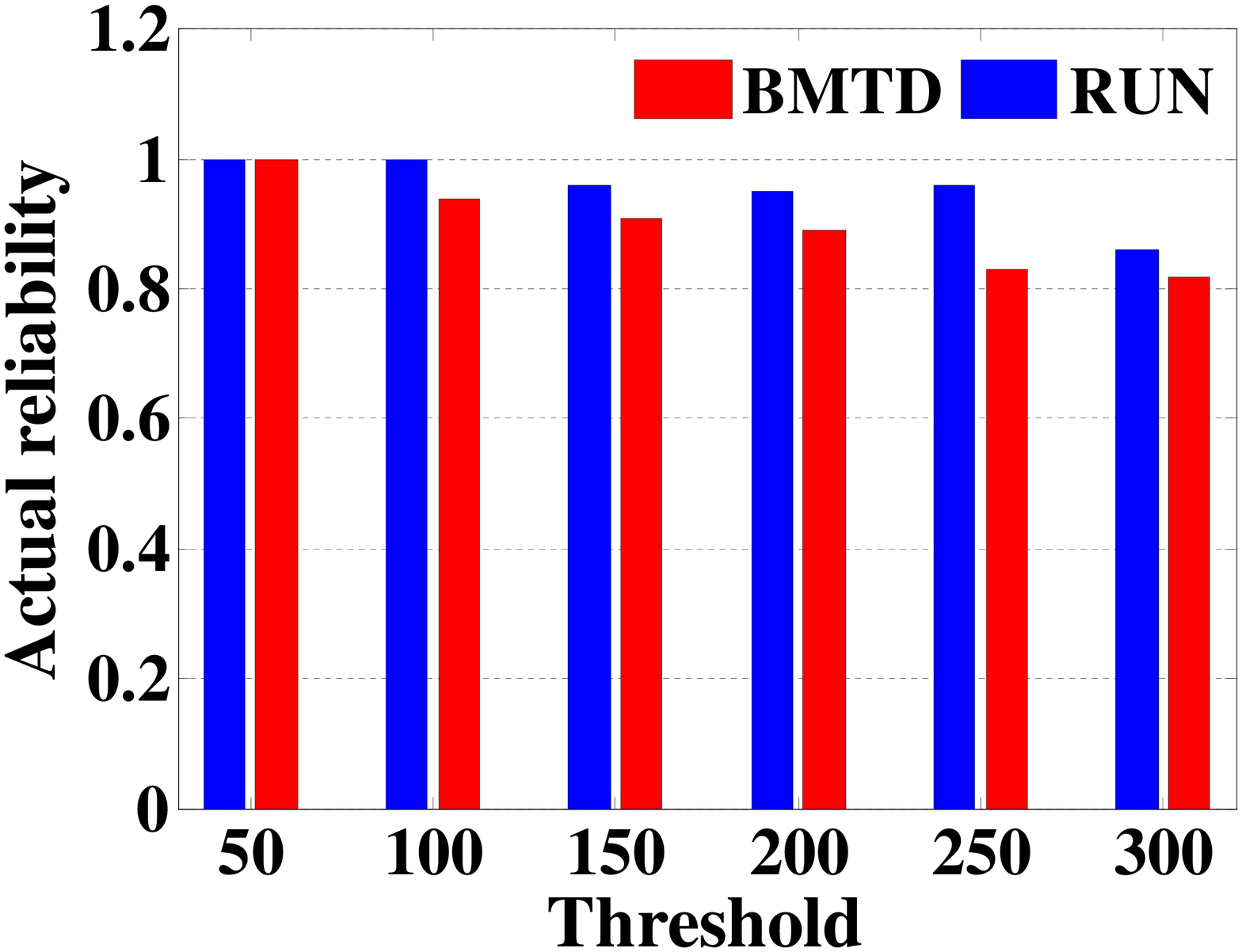}
\label{Fig:R_09_M}
\end{minipage}}
\subfigure[$\alpha=0.99$]{
\begin{minipage}[t]{0.46\linewidth}
\centering
\includegraphics[width=1\textwidth]{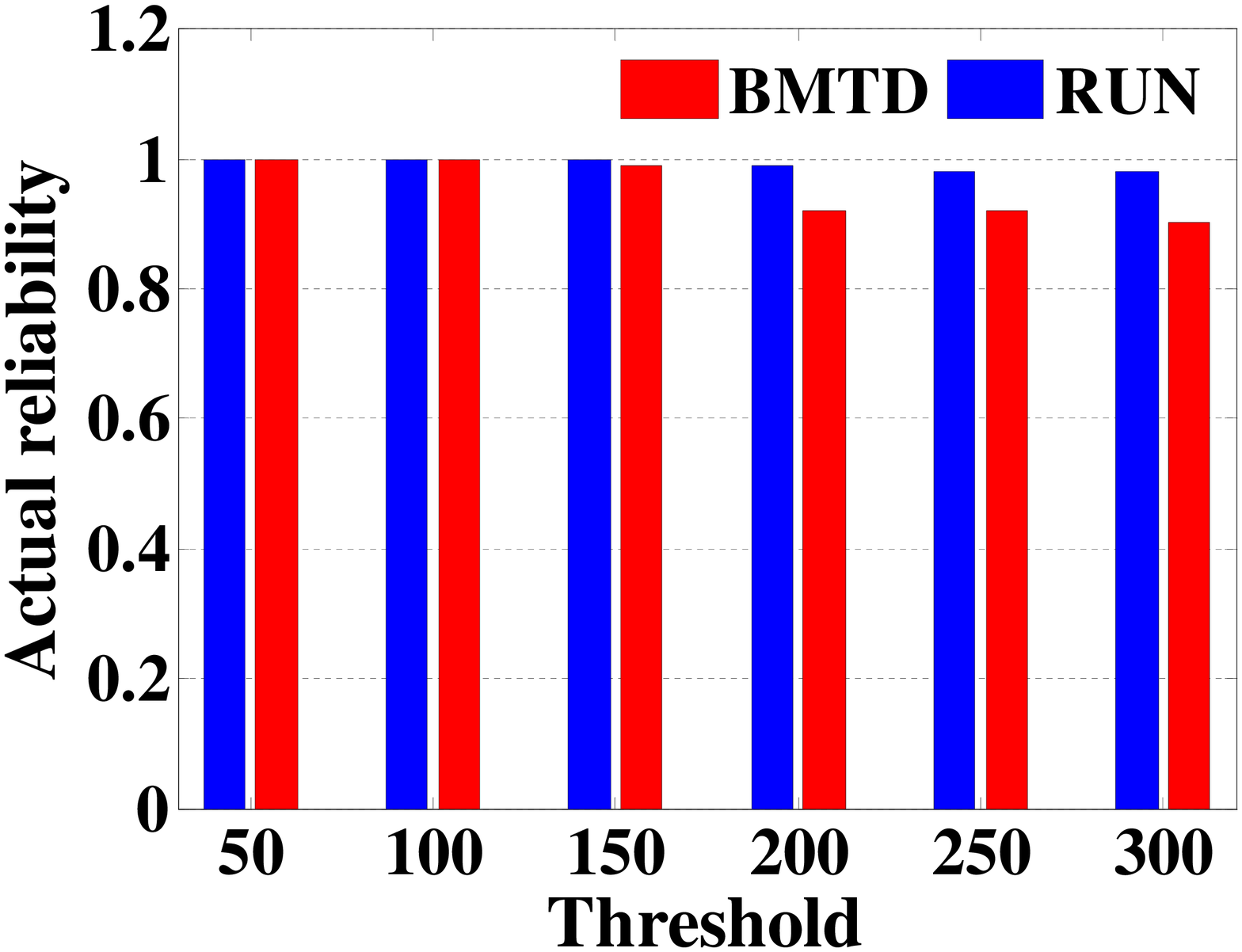}
\label{Fig:R_099_M}
\end{minipage}}
\caption{Actual reliability vs. threshold}
\label{Fig:R_M}
\end{figure}

\begin{figure}[htbp]
\centering
\subfigure[$\alpha=0.9$]{
\begin{minipage}[t]{0.46\linewidth}
\centering
\includegraphics[width=1\textwidth]{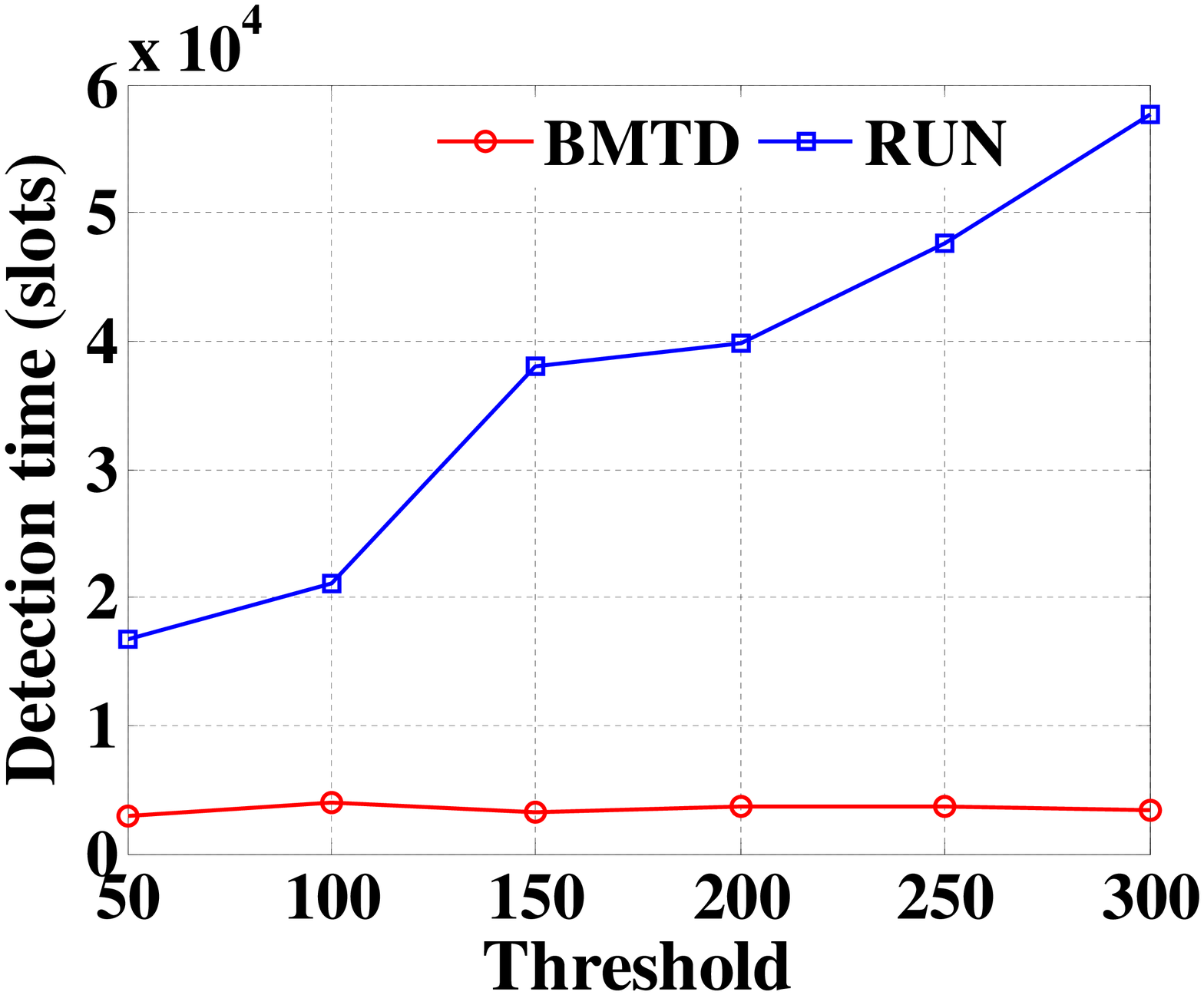}
\label{Fig:T_09_M}
\end{minipage}}
\subfigure[$\alpha=0.99$]{
\begin{minipage}[t]{0.46\linewidth}
\centering
\includegraphics[width=1\textwidth]{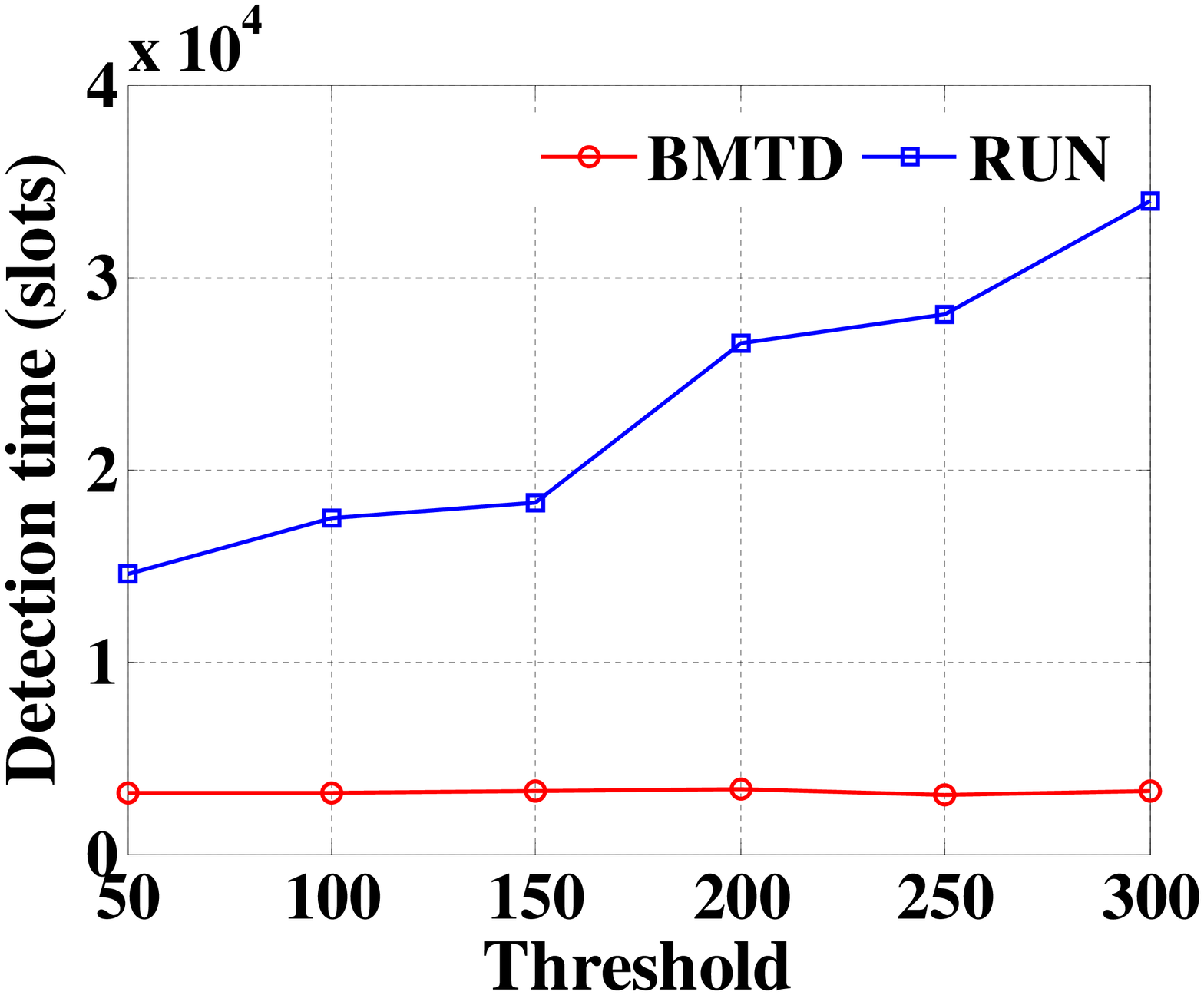}
\label{Fig:T_099_M}
\end{minipage}}
\caption{Detection time vs. threshold}
\label{Fig:T_M}
\end{figure}

\section{Conclusions}
\label{sec:conclusion}
This paper has investigated an important problem of detecting missing tags in the presence of a large number of unexpected tags in large-scale RFID systems. Specifically, we aim at detecting a missing tag event in a reliable and time-efficient way. This paper has proposed a two-phase Bloom filter-based missing tag detection protocol (BMTD). In the first phase, we employed Bloom filter to screen out and then deactivate the unexpected tags in order to reduce their interference to the detection. In the second phase, we further used Bloom filter to test the membership of the expected tags to detect missing tags. We also showed how to configure the protocol parameters so as to optimize the detection time with the required reliability. Furthermore, we conducted extensive simulation experiments to evaluate the performance of the proposed protocol and the results demonstrate the effectiveness and efficiency of the propose protocol in comparison with the state-of-the-art solution.

\bibliographystyle{abbrv}
\bibliography{missing_unknown}

\begin{thebibliography}{10}

\bibitem{bloom1970space}
B.~H. Bloom.
\newblock Space/time trade-offs in hash coding with allowable errors.
\newblock {\em Communications of the ACM}, 13(7):422--426, 1970.

\bibitem{bu2012misplaced}
K.~Bu, B.~Xiao, Q.~Xiao, and S.~Chen.
\newblock Efficient misplaced-tag pinpointing in large {RFID} systems.
\newblock {\em IEEE Transactions on Parallel and Distributed Systems},
  23(11):2094--2106, 2012.

\bibitem{chen2013understanding}
B.~Chen, Z.~Zhou, and H.~Yu.
\newblock Understanding {RFID} counting protocols.
\newblock In {\em ACM MobiHoc}, pages 291--302. ACM, 2013.

\bibitem{chen2011efficient}
S.~Chen, M.~Zhang, and B.~Xiao.
\newblock Efficient information collection protocols for sensor-augmented
  {RFID} networks.
\newblock In {\em IEEE INFOCOM}, pages 3101--3109. IEEE, 2011.

\bibitem{C1G22005}
{EPCglobal Inc.}
\newblock Radio-frequency identity protocols class-1 generation-2 {UHF} {RFID}
  protocol for communications at 860 mhz - 960 mhz version 1.0.9.
\newblock [Online], 2005.
\newblock Available:
  \url{http://www.gs1.org/gsmp/kc/epcglobal/uhfc1g2/uhfc1g2_1_0_9-standard-20050126.pdf}.

\bibitem{han2010counting}
H.~Han, B.~Sheng, C.~C. Tan, Q.~Li, W.~Mao, and S.~Lu.
\newblock Counting {RFID} tags efficiently and anonymously.
\newblock In {\em IEEE INFOCOM}, pages 1--9. IEEE, 2010.

\bibitem{han2014twins}
J.~Han, C.~Qian, X.~Wang, D.~Ma, J.~Zhao, P.~Zhang, W.~Xi, and Z.~Jiang.
\newblock Twins: Device-free object tracking using passive tags.
\newblock In {\em IEEE INFOCOM}, pages 469--476. IEEE, 2014.

\bibitem{hao2007building}
F.~Hao, M.~Kodialam, and T.~Lakshman.
\newblock Building high accuracy bloom filters using partitioned hashing.
\newblock In {\em ACM SIGMETRICS}, pages 277--288. ACM, 2007.

\bibitem{kodialam2007anonymous}
M.~Kodialam, T.~Nandagopal, and W.~C. Lau.
\newblock Anonymous tracking using {RFID} tags.
\newblock In {\em IEEE INFOCOM}, pages 1217--1225. IEEE, 2007.

\bibitem{la2011anticollision}
T.~F. La~Porta, G.~Maselli, and C.~Petrioli.
\newblock Anticollision protocols for single-reader {RFID} systems: temporal
  analysis and optimization.
\newblock {\em IEEE Transactions on Mobile Computing}, 10(2):267--279, 2011.

\bibitem{lee2008supply}
C.-H. Lee and C.-W. Chung.
\newblock Efficient storage scheme and query processing for supply chain
  management using {RFID}.
\newblock In {\em ACM SIGMOD}, pages 291--302. ACM, 2008.

\bibitem{li2010identifying}
T.~Li, S.~Chen, and Y.~Ling.
\newblock Identifying the missing tags in a large {RFID} system.
\newblock In {\em ACM MobiHoc}, pages 1--10. ACM, 2010.

\bibitem{liu2014query}
J.~Liu, B.~Xiao, K.~Bu, and L.~Chen.
\newblock Efficient distributed query processing in large rfid-enabled supply
  chains.
\newblock In {\em IEEE INFOCOM}, pages 163--171. IEEE, 2014.

\bibitem{liu2015completely}
X.~Liu, K.~Li, G.~Min, Y.~Shen, A.~X. Liu, and W.~Qu.
\newblock Completely pinpointing the missing {RFID} tags in a time-efficient
  way.
\newblock {\em IEEE Transactions on Computers}, 64(1):87--96, 2015.

\bibitem{luo2012probabilistic}
W.~Luo, S.~Chen, T.~Li, and Y.~Qiao.
\newblock Probabilistic missing-tag detection and energy-time tradeoff in
  large-scale {RFID} systems.
\newblock In {\em ACM MobiHoc}, pages 95--104. ACM, 2012.

\bibitem{luo2014missing}
W.~Luo, S.~Chen, Y.~Qiao, and T.~Li.
\newblock Missing-tag detection and energy--time tradeoff in large-scale {RFID}
  systems with unreliable channels.
\newblock {\em IEEE/ACM Transactions on Networking}, 22(4):1079--1091, 2014.

\bibitem{mitzenmacher2005probability}
M.~Mitzenmacher and E.~Upfal.
\newblock {\em Probability and computing: Randomized algorithms and
  probabilistic analysis}.
\newblock Cambridge University Press, 2005.

\bibitem{myung2006adaptive}
J.~Myung and W.~Lee.
\newblock Adaptive splitting protocols for {RFID} tag collision arbitration.
\newblock In {\em ACM MobiHoc}, pages 202--213. ACM, 2006.

\bibitem{namboodiri2010energy}
V.~Namboodiri and L.~Gao.
\newblock Energy-aware tag anticollision protocols for rfid systems.
\newblock {\em IEEE Transactions on Mobile Computing}, 9(1):44--59, 2010.

\bibitem{National2015}
{National Retail Federation}.
\newblock National retail security survey.
\newblock [Online], 2015.
\newblock Available:
  \url{https://nrf.com/resources/retail-library/national-retail-security-survey-2015}.

\bibitem{ni2011tracking}
L.~M. Ni, D.~Zhang, and M.~R. Souryal.
\newblock {RFID}-based localization and tracking technologies.
\newblock {\em IEEE Wireless Communications}, 18(2):45--51, 2011.

\bibitem{qian2011cardinality}
C.~Qian, H.~Ngan, Y.~Liu, and L.~M. Ni.
\newblock Cardinality estimation for large-scale {RFID} systems.
\newblock {\em IEEE Transactions on Parallel and Distributed Systems},
  22(9):1441--1454, 2011.

\bibitem{qiao2011polling}
Y.~Qiao, S.~Chen, T.~Li, and S.~Chen.
\newblock Energy-efficient polling protocols in {RFID} systems.
\newblock In {\em ACM MobiHoc}, page~25. ACM, 2011.

\bibitem{DoD2004}
{RFID Journal}.
\newblock {DoD} releases final {RFID} policy.
\newblock [Online], 2004.
\newblock Available:
  \url{http://www.rfidjournal.com/article/articleview/1080/1/1}.

\bibitem{DoD2007}
{RFID Journal}.
\newblock {DoD} reaffirms its {RFID} goals.
\newblock [Online], 2007.
\newblock Available:
  \url{http://www.rfidjournal.com/article/articleview/3211/1/1}.

\bibitem{shahzad2012everybit}
M.~Shahzad and A.~X. Liu.
\newblock Every bit counts: fast and scalable {RFID} estimation.
\newblock In {\em ACM Mobicom}, pages 365--376, 2012.

\bibitem{shahzad2013probabilistic}
M.~Shahzad and A.~X. Liu.
\newblock Probabilistic optimal tree hopping for {RFID} identification.
\newblock In {\em ACM SIGMETRICS}, volume~41, pages 293--304. ACM, 2013.

\bibitem{shahzad2015expecting}
M.~Shahzad and A.~X. Liu.
\newblock Expecting the unexpected: Fast and reliable detection of missing
  {RFID} tags in the wild.
\newblock In {\em IEEE INFOCOM}, pages 1939--1947. IEEE, 2015.

\bibitem{sheng2008finding}
B.~Sheng, C.~C. Tan, Q.~Li, and W.~Mao.
\newblock Finding popular categories for {RFID} tags.
\newblock In {\em ACM MobiHoc}, pages 159--168. ACM, 2008.

\bibitem{tan2008monitor}
C.~C. Tan, B.~Sheng, and Q.~Li.
\newblock How to monitor for missing {RFID} tags.
\newblock In {\em IEEE ICDCS}, pages 295--302. IEEE, 2008.

\bibitem{yang2013localization}
P.~Yang, W.~Wu, M.~Moniri, and C.~C. Chibelushi.
\newblock Efficient object localization using sparsely distributed passive
  {RFID} tags.
\newblock {\em IEEE Transactions on Industrial Electronics}, 60(12):5914--5924,
  2013.

\bibitem{zhang2011fast}
R.~Zhang, Y.~Liu, Y.~Zhang, and J.~Sun.
\newblock Fast identification of the missing tags in a large {RFID} system.
\newblock In {\em IEEE SECON}, pages 278--286. IEEE, 2011.

\bibitem{zheng2013zoe}
Y.~Zheng and M.~Li.
\newblock Zoe: Fast cardinality estimation for large-scale {RFID} systems.
\newblock In {\em IEEE INFOCOM}, pages 908--916. IEEE, 2013.

\end{thebibliography}
\end{document}